%% file: main.tex
\documentclass{article}
\usepackage{graphicx} 
\usepackage{float} 
\usepackage{amsmath,amsthm,amssymb}
\usepackage{comment}

\usepackage{geometry}
\usepackage{parskip}
\usepackage{natbib}
\usepackage{enumitem}
\usepackage{mathtools}
\mathtoolsset{showonlyrefs}
\allowdisplaybreaks

\usepackage{color-edits}
\addauthor[Ying]{ying}{blue}
\addauthor[Dominik]{dominik}{red}
\newtheorem{theorem}{Theorem}

\newtheorem{proposition}{Proposition}
\newtheorem{corollary}{Corollary}

\usepackage{ying}

\usepackage[colorlinks,
            linkcolor=red,
            anchorcolor=blue,
            citecolor=blue
            ]{hyperref}

\title{Augmented Inverse Hybrid Weighting: Robust Inference under Deterministic and Random Distribution Shifts}
\author{Ying Jin and Dominik Rothenh\"ausler}

\begin{document}

\maketitle
\begin{abstract}
Reweighting source samples to match a target covariate distribution is a standard response to distribution shift when generalizing evidence across populations. It is well suited to deterministic, learnable covariate discrepancies, but can fail when source--target differences extend beyond covariate shift or when estimation of the density-ratio weights is unstable.
We introduce a new model where the discrepancy between two population laws comprises a deterministic covariate shift and a non-systematic component, modeled as random perturbation of the probability space. This separates systematic shifts, treated as bias and corrected by reweighting, from random perturbations, treated as distributional uncertainty and reduced by dataset pooling.
Under pure random perturbations, this principle yields Augmented Inverse Distance Weighting, which combines regression augmentation with variance-optimal dataset-level pooling and attains asymptotic variance no larger than exact covariate balancing.
When both components are present, we develop Augmented Inverse Hybrid Weighting, which interpolates between AIDW and standard augmented importance weighting.
Both estimators tradeoff sampling and distributional uncertainty via a \emph{distributional distance} quantifying the strength of random perturbations.
We establish asymptotic normality of both estimators and provide tuning and diagnostic tools.
On three real-world multi-site datasets exhibiting diverse distribution shifts, the proposed methods consistently reduce mean-squared error compared with weighting baselines, and improve empirical coverage when covariate-shift adjustment alone undercovers.
\par\smallskip\noindent\textit{Keywords}: Distribution shift; External validity; Generalizability; Distributional uncertainty; Replication studies.
\end{abstract}


\input{01intro}

\input{02setup}
\input{03aidw}
\input{04aihw}

\input{05implement}

\input{06real}
\input{07simu}

\section{Discussion}\label{sec:discussion}

This paper introduces a new framework for robust inference under distribution shift.
We propose
separating the systematic part of the shift, which can be corrected by
reweighting, from the non-systematic part,
which acts as variance and is better handled by
pooling datasets.
In  real-world multi-site data, this framework yields accurate estimation and reliable uncertainty quantification, outperforming weighting baselines in diverse scenarios.

Several directions remain open. Procedures with end-to-end guarantees for selecting the set $\mathcal{D}$  would further enhance the practical utility of our method, particularly for high-dimensional settings where manual covariate selection becomes difficult.
The model
assumes that a single $\delta_{\mathrm{dist}}^2$
governs covariates and outcomes alike; when
outcome-relevant shifts are larger, as the ACS results
may indicate, a sensitivity analysis that inflates
$\delta_{\mathrm{dist}}^2$ for the outcome would be a
natural extension.  Finally, the framework could be
extended beyond the mean to quantiles and parameters in empirical risk minimization and different asymptotic regimes of sample sizes and number of perturbation cells.

\section*{Acknowledgments}

Rothenh\"ausler gratefully acknowledges support as a David Huntington Faculty Scholar, Chamber Fellow,
and from the Dieter Schwarz Foundation.

\bibliographystyle{plainnat}
\bibliography{reference}

\appendix
\input{10appendix}

\end{document}

%% file: 01intro.tex

\section{Introduction}
 
Distribution shift is a central challenge when data from one population are used to support conclusions about a related target population~\citep{degtiar2023review}. 
This problem is common in clinical trials~\citep{deaton2018understanding,stuart2011use}, social science~\citep{hotz2005predicting}, and machine learning~\citep{quinonero2008dataset}, where outcomes may be observed in a source population but unavailable in the target. 
For example, in evaluating a new treatment, we may observe the participants' background and outcomes from a source hospital, but only limited covariates from a target hospital where the treatment is yet to roll out. The key challenge in estimating estimate the target mean outcome is to account for differences between the two populations.

A standard response is to assume covariate shift, that is,  the discrepancy between the two populations is fully explained by the covariates. The target estimand can then be recovered by estimating the density ratio between the covariate distributions and reweighting source observations to resemble the target population~\citep{shimodaira2000improving,horvitz1952generalization,robins1994estimation}. This solution is natural and powerful when the distribution shift is systematic, stable, and well captured by observed covariates. 

However, recent empirical studies in machine learning, causal inference, and replication studies consistently documented that substantial cross-population variation often remains after adjusting for observed covariates~\citep{cai2023diagnosing,lu2023you,jin2023diagnosing}. In a large-scale multi-site evaluation, uncertainty quantification based only on covariate-shift adjustment can substantially undercover~\citep{jin2024beyond}. 
These results suggest that residual shifts can remain beyond the discrepancies in the covariates. 

Importantly, this residual variation need not be entirely unstructured. 
Several empirical analyses found that, after systematic covariate differences are removed, the remaining cross-site discrepancies are approximately non-systematic, without a persistent direction across populations~\citep{jin2024beyond,jeong2024out}. This suggests that part of the residual shift may be better viewed as a realization of a random perturbation than as a deterministic discrepancy to be estimated and removed. An analysis
in \citet{jin2024beyond} provided preliminary evidence that exploiting this
random-shift structure can improve effect transportation and uncertainty
quantification, suggesting that the likely stochastic structure of this residual shift can be statistically useful. However, a principled framework for estimation and inference under such random shifts or, more realistically, a mixture of both types, remains missing. 
In this paper, we formalize a random-perturbation perspective for residual distribution shift---first on its own and then alongside systematic covariate shift---and use it as a working model to develop methods that exploit this structure when it is (approximately) correct, while remaining robust for real datasets under  complex distribution shifts.

\subsection{Modeling deterministic and random shifts}
\label{subsec:intro_model}

We develop a  model that separates systematic discrepancies from random perturbations. Let $P_s$ denote the source law and $P_t$ the target. We assume
\@\label{eq:shift_high_level}
P_s
\ \xrightarrow{\text{deterministic, learnable shift}}\
P'
\ \xrightarrow{\text{centered random perturbation}}\
P_t,
\@
where \(P'\) is an intermediate population obtained after applying a
stable component to $P_s$.  
The concrete definitions are in Section~\ref{sec:notation}. 
The two components serve distinct roles:  

\begin{itemize}
\item The first \textbf{systematic} component is modeled through \emph{observed covariates}. It resembles the standard covariate shift perspective, which captures stable, learnable differences. 
It may arise from stable differences in site populations, sampling criteria set by investigators, or site specialization, e.g., a target hospital serving older or more severe patients.  
\item 
The second is the residual component that remains after accounting for systematic differences, which we model as \textbf{random} and centered around $P'$.  
We shall let it model the aggregate effect of many small, non-systematic factors (such as local recruitment fluctuations, referral patterns, operational variation, routine site-specific differences) that happen randomly to a specific target. It reflects the empirical insights above. 
\end{itemize}

Under this model, the discrepancy between $P_s$ and $P_t$ may contain a realization of a random perturbation which affects both the covariates and the unknown outcomes. 
Thus, fitting the covariate density-ratio suffers from two issues: (i) it may chase random fluctuations in the covariates that inflates estimation error, and (ii) additional variation in the outcomes still remains. 
The central estimation challenge is therefore to exploit any deterministic 
transport structure without overfitting to the random fluctuations. 

From a modeling perspective, both covariate shift and pure random shift are useful endpoints, and in many applications, one may expect both to be present. We  study the random shift setting first ($P'=P_s$) and then the general hybrid setting with both types.

\subsection{Overview of methodological contributions}

Our methods address the above challenge by assigning different statistical
roles to the two components. We use \emph{reweighting} to address the bias (deterministic shift) following the conventional wisdom in the  literature~\citep{robins1994estimation}. 
By contrast, the variance part (random shift) is addressed through \emph{pooling}  datasets.  

We first study the pure random-shift model, which leads to an Augmented Inverse Distance Weighting (AIDW) estimator. AIDW pools source and target covariate information in the regression adjustment, thereby trading off the sampling variability and the distributional discrepancy. Under the random-shift model, we show that AIDW remains unbiased for any pooling level, and its asymptotic variance is minimized by an optimal pooling rate that depends on the strength of the random perturbation.  
In addition, we show that exact covariate balancing, a deliberately favorable benchmark for reweighting, has asymptotic variance no smaller than that of oracle AIDW, with a strictly positive gap whenever the balancing weights are non-constant.

We then extend the framework to the hybrid setting and develop Augmented Inverse Hybrid Weighting (AIHW), which combines reweighting  with pooling. 
It interpolates between AIDW, which treats the discrepancy as purely random, and standard augmented importance weighting, which treats the full covariate discrepancy as deterministic. 
We derive closed-form asymptotic inferential guarantees with variance formula that indicates the optimal choice of the pooling strategy. 
Furthermore, we provide plug-in guidance for the pooling parameters as well as practical tools for estimating the systematic shift component and assessing the hybrid shift model.

Table~\ref{tab:method-comparison} provides a high-level comparison of our proposed methods with existing approaches, summarizing their view of distribution shift and estimation strategies.

\begin{table}[h]
\centering
\begin{tabular}{@{}l l l@{}}
\toprule
Method & Shift model & Estimation strategy \\
\midrule
AIPW & Deterministic & Reweighting with outcome model \\
Covariate balancing & Deterministic & Weights balancing chosen features \\
AIDW (ours) & Random & Pooling across datasets \\
AIHW (ours) & Deterministic and random & Reweighting combined with pooling \\
\bottomrule
\end{tabular}
\caption{Shift models and estimation strategies of the compared methods.}
\label{tab:method-comparison}
\end{table}

We demonstrate the efficacy of the proposed methods via observed target-site benchmarks across three real-world multi-site datasets. In Figure~\ref{fig:intro}, our methods consistently deliver lower mean-squared error  and higher empirical coverage than the covariate-shift-based baselines  across diverse distribution shifts (we expect the datasets to exhibit random, hybrid, and unknown, complicated shifts, respectively; see Section~\ref{sec:experiments}). 

\begin{figure}[H]
    \centering
    \includegraphics[width=0.9\linewidth]{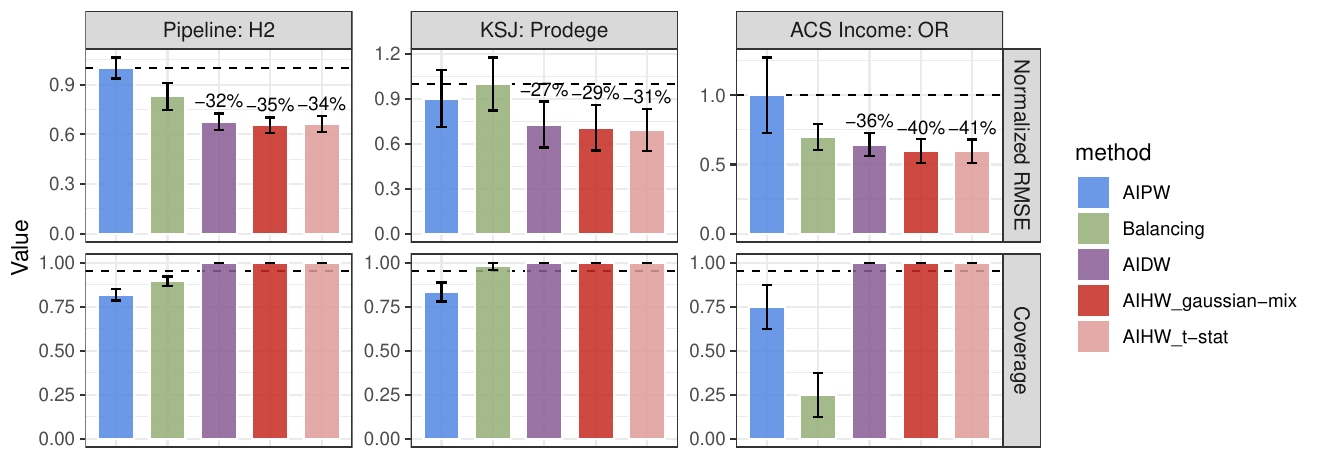}
    \caption{{Preview of empirical results across three real datasets in Section~\ref{sec:experiments} (Hypothesis 2 in Pipeline, Prodege as target site in KSJ, and OR as target site in ACS), where our methods achieve low RMSE and improved uncertainty quantification compared to covariate-shift-based methods. The first row shows the RMSE of the transport estimator from the full-observation target estimator, and the second row shows the empirical coverage of the corresponding prediction intervals. AIPW and Balancing are two covariate-shift based methods; AIDW is our method for  random shift;  two AIHW methods are our methods for hybrid shift, implemented with two procedures for estimating the deterministic component (gaussian-mixture and t-stat).}}
    \label{fig:intro}
\end{figure}


We close this section by a review  of related work. The rest of the paper is organized as follows. Section~\ref{sec:notation} introduces problem setup and the distribution shift models. Section~\ref{sec:aidw} presents AIDW for purely random perturbations. Section~\ref{sec:aihw} develops AIHW for hybrid shift. Section~\ref{sec:practical} discusses practical implementation. Section~\ref{sec:experiments} provides empirical validation. Section~\ref{sec:discussion} concludes with limitations and future directions.

\subsection{Related work}\label{sec:related-work}

This work sits in the intersection of generalizability, external validity, and distribution shift. 
Generalizing experimental findings to new populations
is a long-standing concern in the medical and social
sciences \citep{shadish2002,imai2008misunderstandings,deaton2018understanding,Tipton:2013ew,bareinboim2016fusion}. This perspective is also related to (multi-site) replication studies that investigate how variability in covariates affect effect discrepancies across sites~\citep{klein2014investigating,coppock2018,mcshane2022modeling,holzmeister2024heterogeneity}.
In effect generalization, the prevailing approach assumes that
observed covariates capture all relevant differences
between populations and adjusts for them by reweighting
or outcome modelling, building on the Horvitz-Thompson estimator~\citep{horvitz1952generalization}. In modern practice it often takes the form of inverse-probability/density-ratio weights. The instability of raw importance weights motivated variance-reduction strategies, including augmentation with outcome models~\citep{robins1994estimation} and balancing  that regularizes the weights~\citep{deville1992, hainmueller2012entropy,zubizarreta2015stable}. Such methods are now widely used to generalize experiments across sites~\citep{cole2010generalizing, stuart2011use, Tipton:2013ew, hartman2015sate, buchanan2018ipsw, dahabreh2019generalizing, egami2021covariate}. Our AIHW estimator builds on a similar augmentation idea but departs from standard practice by explicitly reducing variance due to random perturbations.

The domain-adaptation literature also centers on importance weighting and related estimators~\citep{shimodaira2000improving,bickel2007discriminative}. This line typically assumes a deterministic change in covariate distribution (see survey in~\cite{quinonero2008dataset, pan2009survey}). Our deterministic-plus-random-perturbation model extends this paradigm by decomposing distributional shift into a deterministic component and random perturbations, then deriving tailored reweighting procedures.


Random-perturbation models were introduced to calibrate statistical inference under excess distributional variation~\citep{jeong2022calibrated} and empirical risk minimization~\citep{jeong2024out}. Practical relevance has been demonstrated in applications such as refugee assignment~\citep{bansak2024learning} and replication datasets~\citep{jin2024beyond}. This paper complements and extends that literature by decomposing shift into deterministic shifts and random perturbations and by deriving hybrid weighting estimators and diagnostics targeted at robust estimation and inference.

%% file: 02setup.tex
\vspace{-0.5em}
\section{Problem setup and distribution shift models}\label{sec:notation}

Let $X\in \RR^d$ denote covariates and $Y\in \RR$ an outcome. 
We observe i.i.d.\ source data $\{(X_i,Y_i)\}_{i=1}^{n_s}$ from a source distribution $P_s$ and i.i.d.\ target covariates $X_i'$, $i=1,\ldots,n_t$, drawn from the covariate marginal induced by a target law $P_t$ on $(X,Y)$. 
The target data is from $\{(X_i',Y_i')\}_{i=1}^{n_t}\iid P_t$. To simplify the discussion, we focus on outcome mean estimation. For average treatment effect in randomized experiments, one could apply the methods in treated and control groups separately; see Appendix~\ref{app:extend_linear}.  

Most distribution-shift models make assumptions on the unknown, \emph{fixed} density ratio $\ud P_t/\ud P_s$. For example, $f$-divergence balls restrict the size of this density ratio~\citep{bental2013robust, duchi2021learning}, and covariate-shift models assume it is a function of $X$ only~\citep{shimodaira2000improving, bickel2007discriminative}. 
We instead allow part of this ratio to be random, reflecting the non-systematic factors that generate a particular target population. In words, following the picture $P_s\to P'\to P_t$ in~\eqref{eq:shift_high_level}, the target population
arises from the source in two steps. A systematic step ($P_s\to P'$) changes the
distribution of a subset of covariates $X_\cD$, as when a hospital
specializes in older patients. A random step ($P'\to P_t$) then makes many small
regions of the joint $(X,Y)$ space slightly more or less common, such as
when referral patterns or local practices happen to differ, with no
direction favoured in advance.
We present such a hybrid model~\eqref{eq:shift_high_level} directly, so the pure random perturbation is a special case without the systematic component.

Let $(I_j)_{j=1}^J$ be a partition of the joint covariate--outcome space, and let \((W_j)_{j=1}^J\) be positive random variables, with $\bar W:=J^{-1}\sum_{j'=1}^J W_{j'}$. 
We model the realized likelihood ratio between target and source as the product of a fixed factor and a random factor:
\begin{equation}\label{eq:hybrid}
  \frac{dP_t(\cdot\,;W_\bullet)}{dP_s}(x,y)
  = w_\cD(x_\cD)\,\frac{W_j}{\bar W},
  \qquad (x,y)\in I_j .
\end{equation}
The fixed component $w_\cD(x_\cD)$ is an ordinary covariate-shift density ratio on a subset of covariates $\cD\subseteq\{1,\dots,d\}$, with $\mathbb{E}_s[w_\cD(X_\cD)]=1$. It defines the intermediate law $dP_{s,\cD}(x,y):=w_\cD(x_\cD)\,dP_s(x,y)$, the $P'$ of~\eqref{eq:shift_high_level}, under which we write $\mathbb{E}_{s,\cD}$ and $\mathrm{Var}_{P_{s,\cD}}$. 
The random factor $W_j/\bar W$ rescales the
mass of cell $I_j$: if a cell contains, e.g., older patients with
severe symptoms and good outcomes, $W_j/\bar W>1$ means that this
group happens to be over-represented in the target. Because the cells partition $(X,Y)$ jointly, the perturbation changes both the covariate distribution and the conditional law of $Y\given X$. In the hospital example, this can capture how unobserved factors like varying referral patterns randomly perturb the patient mixture and outcome behavior. 
Under this model, $W_j\equiv1$ gives covariate shift with density ratio $w_\cD(X_\cD)$, and $\cD=\emptyset$ (so $w_\cD\equiv1$) gives the pure random-perturbation model, which we study first in Section~\ref{sec:aidw}.


\noindent\textbf{Inferential target.}
In this work, the estimand is $\theta=\theta(W_\bullet):=\mathbb{E}_{P_t(\cdot;W_\bullet)}[Y]$, which is random through the draw of the target population. All our results concern the marginal error $\hat\theta-\theta(W_\bullet)$ over the perturbation, source sampling and target sampling.

Throughout, we make the following regularity conditions, under which the perturbations $\{W_j\}$ are i.i.d., the cells are increasingly fine-grained, and the sample sizes grow proportionally, so that the distributional uncertainty is statistically tractable. 

\begin{assumption}\label{ass:model}
(a) The weights $W_1,\dots,W_J$ are i.i.d.\ whose distribution does not depend on $J$, with $\mathbb{E}[W_1]=1$, $\mathrm{Var}(W_1)<\infty$ and $W_1\ge c$ for a constant $c>0$.
(b) $c_w\le w_D(x_D)\le C_w$ for constants $0<c_w\le C_w<\infty$.
(c) The cells have equal mass, $P_{s,D}(I_j)=1/J$ for all $j$, and step functions on $\{I_j\}_{j=1}^J$ are dense in $L^2(P_{s,D})$ as $J\to\infty$.
(d) Conditional on $W_\bullet$, the source sample is i.i.d.\ from $P_s$, the full target sample is i.i.d.\ from $P_t(\cdot\,;W_\bullet)$, and the two are independent.
(e) As $n\to \infty$, $n_s/J\to\rho_s\in(0,\infty)$ and $n_t/J\to\rho_t\in(0,\infty)$. 
\end{assumption}

The key implication of our model (Assumption~\ref{ass:model})---which  we shall leverage to derive our methodology and theory---is a distributional central limit theorem (CLT) that quantifies how the random perturbation affects the source-target discrepancies.
%
In specific, for pure random shift ($\cD=\emptyset$), the distributional CLT (Theorem~\ref{thm:matched-clt} in Appendix~\ref{subsec:distributional-clt-matched}) states that, for any fixed $f\in L^2(P_s)$, we have $\frac{\hat{\mathbb{E}}_s[f(X,Y)]-\hat{\mathbb{E}}_t[f(X,Y)]}{s_n}\xrightarrow{d}N(0,1)$ where
\begin{equation}\label{eq:dclt}
  s_n^2=\Big(\underbrace{1/n_s}_{\text{source sampling}}
  +\underbrace{1/n_t}_{\text{target sampling}}
  +\underbrace{\delta_{\mathrm{dist}}^2}_{\text{distributional}}\Big)\,
  \mathrm{Var}_{P_s}\{f(X,Y)\}, 
\end{equation}
where $\hat{\mathbb{E}}_s$ and $\hat{\mathbb{E}}_t$ are empirical means over the source and the full target samples, and 
\begin{equation}\label{eq:delta_dist}
  \delta_{\mathrm{dist}}^2 := \frac{\mathrm{Var}(W_1)}{J} 
\end{equation}
is the \textbf{distributional distance} that summarizes the strength of the random component.
Thus $\delta_{\mathrm{dist}}^2$ measures the contribution of the random perturbation to source--target discrepancies on the same scale as the sampling variances.  
Related distributional CLT results for  random shift models have been developed in various asymptotic regimes~\citep{jeong2024out,zhang2025data}. 
Theorem~\ref{thm:matched-clt}
extends them by giving the joint limit of the empirical
means and the randomly perturbed population means,
which we need for the random estimand $\theta(W_\bullet)$, and
Corollary~\ref{cor:hybrid_phi_discrepancy} extends it
to the hybrid model. These results serve  as the basis for the paper by characterizing the sources of uncertainty that our estimators trade off (Sections~\ref{sec:aidw} and~\ref{sec:aihw}).  Building on them, we
establish asymptotic normality of our estimators with
estimated nuisance functions, characterize the
variance-optimal pooling, compare it with exact
covariate balancing, and derive estimation and diagnostic tools.

The distributional CLT reveals  that the distributional distance $\delta_{\textnormal{dist}}^2$ alone controls the discrepancy of every function of $(X,Y)$. Intuitively, this allows $\delta_{\mathrm{dist}}^2$ to be estimated from covariates only, even though the target outcomes are unobserved (Section~\ref{sec:estimation-of-delta}) so the model can be 
partly checked from data (Section~\ref{section:diagnostics} and
Figure~\ref{fig:real_ksj_diagnostic}), and earlier analyses found residual shifts consistent with this structure~\citep{jeong2022calibrated,jin2024beyond,jeong2024out}. The results in
Section~\ref{sec:experiments}, where target outcomes
are available for evaluation, provide partial
evidence too. 
Finally, the finite partition $(I_j)_{j=1}^J$ can be viewed as a device to construct the target law; readers can view the distributional CLT as the defining feature of the model. Such equal-mass partitions exist for any law including those with discrete covariates or
outcomes, after augmenting $(X,Y)$ with extra randomness
(Appendix~\ref{app:atoms}).
%


%% file: 03aidw.tex

\section{Augmented Inverse Distance Weighting}\label{sec:aidw}

We first study effect generalization under the pure random shift model with $\cD=\emptyset$. 
Under purely random perturbations, reweighting individual units can be suboptimal as it chases unpredictable fluctuations, inflating the variance of the weights. 
We instead treat the shift as a source of
variance and reduce it by pooling the source and target datasets.


\subsection{The AIDW estimator}

Our goal is to estimate the random target mean $\theta=\theta(W_\bullet)=\mathbb{E}_{P_t(\cdot;W_\bullet)}[Y]$ using source samples $\{(X_i,Y_i)\}_{i=1}^{n_s}$ from $P_s$ and target covariates $\{X_i'\}_{i=1}^{n_t}$ from the target distribution $P_t(\cdot;W_\bullet)$.  We propose the Augmented Inverse Distance Weighting (AIDW) estimator:
\begin{equation*}
 \hat \theta_\text{AIDW}(\alpha) = \underbrace{\frac{1}{n_s} \sum_{i=1}^{n_s} (Y_i - \hat Q(X_i))}_{\text{source residual correction}} +  \underbrace{\alpha \frac{1}{n_s} \sum_{i=1}^{n_s} \hat Q(X_i)  + (1-\alpha) \frac{1}{n_t} \sum_{i=1}^{n_t} \hat Q(X_i')}_{\text{dataset pooling}},
\end{equation*}
where $\hat Q(X)$ is an estimate of the conditional mean $\mathbb{E}_s[Y\given X]$. 
Our theory assumes that $\hat Q$ is fitted on held-out data independent of the current samples and of the perturbation.\footnote{In practice, we use standard source-side cross-fitting which is independent of the perturbation. Extending the  theory to standard cross-fitting~\citep{chernozhukov2018double} would require additional bookkeeping, which we do not pursue here.} 


The AIDW estimator consists of two terms. 
The first source residual correction term is similar to that in the AIPW estimator~\citep{robins1994estimation}, which corrects for the estimation error in the outcome model. 
The second dataset pooling term combines outcome predictions from source and target datasets. When $\alpha=1$, the AIDW estimator fully relies on the source data and reduces to the source sample mean without any transfer. When $\alpha=0$, the estimator fully relies on the target covariates and reduces to the AIPW estimator with weights equal to $1$ (i.e., no reweighting).

The distinctive feature of the AIDW estimator is that it is designed for non-systematic shift that acts as variance. We interpret it by its ``unbiasedness'' property. 
Conditional on a realization $W_{\bullet}$ and $\hat{Q}$, AIDW is biased for $\theta(W_\bullet)$: for 
$f_{\alpha,\hat{Q}}(x,y):=y - (1-\alpha)\hat{Q}(x)$,
\$
\mathbb{E}[\hat \theta_{\text{AIDW}}\given W_\bullet,\hat Q] - \theta(W_{\bullet}) = \EE_s[f_{\alpha,\hat{Q}}] - \EE_t[f_{\alpha,\hat{Q}}].
\$
Under the random perturbation, however, this conditional bias has mean zero across perturbation draws: $\EE[\EE_s[f_{\alpha,\hat{Q}}] - \EE_t[f_{\alpha,\hat{Q}}]\given \hat{Q}] = 0$, so the estimator is marginally unbiased and the conditional bias contributes as variance. 
Indeed, at the leading order, this contribution is proportional to 
$\delta_{\text{dist}}^2\Var_{P_s}(Y-(1-\alpha)\hat{Q}(X))$. Thus, the weight $\alpha$ trades this  
variance against target sampling variance, and its optimal value depends on the distributional distance $\delta_\text{dist}^2$ defined in~\eqref{eq:delta_dist}; see Theorem~\ref{theorem:aidw}. 
The name AIDW reflects the role of $\delta_{\text{dist}}^2$ in determining the optimal pooling strategy.


\subsection{Theoretical properties of AIDW}
We establish the theoretical properties of AIDW, including its asymptotic distribution and variance which facilitate comparisons with existing methods and the choice of $\alpha$.

\vspace{0.25em}
\begin{theorem}[AIDW]\label{theorem:aidw}
Suppose Assumptions~\ref{ass:model} holds, $\hat Q$ is fit independent of the observations and $W_{\bullet}$, and obey $\|\hat Q-Q\|_{L^2(P_s)} = o_P(1)$ and $|\hat \alpha - \alpha| = o_P(1)$ for some fixed function $Q \in L^2(P_s)$ and $\alpha \in [0,1]$. Assume $\mathrm{Var}_{P_{s}}(Q)>0$ and $\mathrm{Var}_{P_{s}}(Y - (1-\alpha) Q )>0$. 
Then, we have  
  $s_n^{-1}(\hat \theta_\mathrm{AIDW}(\hat \alpha) - \theta) \rightarrow \mathcal{N}(0, 1)$,
where the randomness is over the sampling process in Assumption~\ref{ass:model} and $\hat{Q}$ and $\hat\alpha$, and we define 
\begin{equation}\label{eq:sn}
    s_n^2 =    \big\{  1/n_s + \delta_\text{dist}^2 \big\}\cdot \mathrm{Var}_{P_{s}}\big(Y - (1-\alpha) Q\big)  +   \frac{(1-\alpha)^2}{n_t} \cdot \mathrm{Var}_{P_{s}}(Q).
\end{equation} 
\end{theorem}
\vspace{0.25em}

The proof is in Appendix~\ref{app:subsec_proof_aidw}.
Here we allow $\hat Q$ to converge to an arbitrary $Q\in L^2(P_s)$: in this case, AIDW stays asymptotically centered at $\theta$, and the estimator remains unbiased, albeit with larger variance if $Q$ is not the conditional mean function. 
When $Q(x)=\mathbb E_s[Y\given X=x]$, a simple calculation yields 
\begin{equation}\label{eq:sn_correct}
s_n^2 =  \big\{  1/n_s + \delta_\text{dist}^2 \big\}\cdot \mathrm{Var}_{P_{s}}(Y - Q)  +   \big\{ \alpha^2 \left( 1/n_s + \delta_\text{dist}^2 \right) + (1-\alpha)^2 /n_t \big\} \cdot \mathrm{Var}_{P_{s}}(Q).
\end{equation}
The variance formula~\eqref{eq:sn_correct} differs from the standard variance of AIPW or difference-in-mean estimators; our analysis relies on decomposing the uncertainty in both i.i.d.~sampling and random perturbations. 
Intuitively, suppose $\hat Q=Q$ (estimation contributes higher-order error), and write 
\( 
\varepsilon(X,Y):=Y-Q(X).
\)
Then this ``oracle'' AIDW estimator satisfies
\begin{align*}
\hat\theta_{\mathrm{AIDW}}^{\mathrm{orc}}(\alpha)-\theta(W_\bullet)
&=
\underbrace{\hat{\mathbb{E}}_s[\varepsilon]-\mathbb{E}_s[\varepsilon]}_{\text{residual sampling}} + 
\underbrace{\mathbb{E}_s[\varepsilon]-\mathbb{E}_t[\varepsilon]}_{\text{residual shift}}
+
\underbrace{\alpha\big(\hat{\mathbb{E}}_s[Q]-\mathbb{E}_t[Q]\big)
+
(1-\alpha)\big(\hat{\mathbb{E}}_t[Q]-\mathbb{E}_t[Q]\big)}_{\text{prediction part}}.
\end{align*}
The first term concerns the sampling error in the source empirical mean, while the second term (population parameter difference) relies on the random perturbation strength. 
These two terms contribute $(1/n_s + \delta_{\mathrm{dist}}^2)\cdot \mathrm{Var}_{P_s}(Y-Q)$. 
Using same argument, the first term in the ``prediction'' part contributes $\alpha^2 ( 1/n_s + \delta_{\mathrm{dist}}^2 )\cdot  \mathrm{Var}_{P_s}(Q)$, where as the second term  involves the target i.i.d.~sampling uncertainty and contributes $(1-\alpha)^2  \mathrm{Var}_{P_s}(Q)/n_t$.

Another interpretation of~\eqref{eq:sn_correct} is to separate the contributions of the usual sampling uncertainty and the distribution uncertainty. 
Under the marginal perspective, the perturbation contributes a variance term  $\delta_{\text{dist}}^2 \cdot \{\Var_{P_s}(Y-Q) + \alpha^2 \Var_{P_s}(Q)\}$, which depends on $\delta_{\text{dist}}^2$ and $\alpha$. Thus, $\alpha$ can be used to trade off the two sources of uncertainty. 

\noindent\textbf{Optimal choice of $\alpha$.}
Minimizing the  variance formula~\eqref{eq:sn_correct} with respect to $\alpha$ yields  
\begin{equation}\label{eq:opt_alpha_aidw}
\alpha^* = \frac{1/(1/n_s + \delta_{\text{dist}}^2)}{n_t + 1/(1/n_s + \delta_{\text{dist}}^2)} = \frac{1/n_t}{1/n_s + \delta_{\text{dist}}^2 + 1/n_t}.
\end{equation}
Intuitively, a larger value of $\delta_{\text{dist}}^2$, i.e., stronger shift, leads to a smaller value of $\alpha^*$ which gives less weights to source data. 
More specifically, $\alpha^*$ balances source sampling uncertainty ($1/n_s$), target sampling uncertainty ($1/n_t$), and distributional uncertainty ($\delta_{\text{dist}}^2$). 
It can be interpreted as the inverse-variance weight via the ``effective sample sizes'' calculated based on the random shift. For estimating the target parameter $\EE_t[Y]$, the source empirical mean $\hat\EE_s[Y]$ has effective sample size $1/(1/n_s+\delta_{\text{dist}}^2)$, while the target empirical mean $\hat\EE_t[Y]$ has effective sample size $n_t$ (such calculation remains the same after incorporating the regression adjustment). 
This optimal choice can be estimated by plugging in an estimate of  $\delta_{\text{dist}}^2$; we discuss the estimation issue in  Section~\ref{sec:estimation-of-delta}.

\subsection{Comparison with exact covariate balancing}\label{sec:limitations_individual_level_reweighting}

We further demonstrate the benefits of dataset pooling by comparing it with methods that weights individual units. 
To keep the comparison transparent, we use a benchmark that favors reweighting: exact balance on a well-specified, finite-dimensional outcome model with homoskedastic noise. Let $b(X) \in \mathbb{R}^p$ be a fixed feature map whose first coordinate is  equal to $1$, and suppose 
$Q_0(X) = \mathbb{E}_s[Y \given X] = \beta_0^\top b(X)$, 
$\varepsilon = Y - Q_0(X)$, 
where, under $P_s$, $\varepsilon$ is independent of $X$, $\mathbb{E}_s[\varepsilon]=0$, and $\mathrm{Var}_{P_s}(\varepsilon)=\sigma_\varepsilon^2>0$. 
Consider
weights $\hat w$ that depend only on the source and target covariates
and achieve exact balance,
\begin{equation}\label{eq:exact_balance_basis}
\hat{\mathbb{E}}_s \left[\hat w(X)b(X)\right]
=
\hat{\mathbb{E}}_t[b(X)]
\end{equation}
and the balancing estimator
$\hat\theta_{\mathrm{bal}}=\hat{\mathbb{E}}_s[\hat w(X)Y]$; under
\eqref{eq:exact_balance_basis}, adding a regression augmentation
$\hat\beta^\top b(X)$ leaves it unchanged for every $\hat\beta$.
Because the first component of $b(X)$ is $1$, exact balance implies $\hat{\mathbb{E}}_s[\hat w(X)] = 1$.

\vspace{0.25em}
\begin{proposition}[Exact-balancing benchmark]\label{prop:exact_balance_benchmark}
Under the  conditions of Theorem~\ref{theorem:aidw}, assume the homoskedastic noise model above, 
$\mathbb{E}_s  [\|b(X)\|_2^2 ] < \infty$,
and 
$\mathbb{E} [\hat{\mathbb{E}}_s [\hat w(X)^2] ] < \infty$. 
If \eqref{eq:exact_balance_basis} holds almost surely, then $\hat\theta_{\text{bal}}$ has zero bias for $\theta(W_\bullet)$, 
with the leading-order variance 
\begin{equation}\label{eq:exact_balance_variance_asymptotic}
s_{\textnormal{bal}}^2 = \mathrm{AVar}(\hat\theta_{\mathrm{bal}}-\theta(W_\bullet))
=
\frac{1}{n_t}\,\mathrm{Var}_{P_s}(Q_0(X))
+
\frac{\sigma_\varepsilon^2}{n_s}\,\mathbb{E} \left[\hat{\mathbb{E}}_s \left[\hat w(X)^2\right]\right]
+
\delta_{\mathrm{dist}}^2 \,\sigma_\varepsilon^2.
\end{equation}
%
Let \(s^2_{\mathrm{AIDW},*}\) be the leading-order variance of AIDW with \(Q=Q_0\) and with \(\alpha=\alpha^*\). Then  
\@\label{eq:aidw_vs_exact_balance_gap} 
s_{\textnormal{bal}}^2 - s^2_{\mathrm{AIDW},*}
=\,\,&
\bigg(
\frac{1}{n_t}
-
\frac{1}{\,n_t + \frac{1}{1/n_s+\delta_{\mathrm{dist}}^2}\,}
\bigg)\mathrm{Var}_{P_s}(Q_0(X)) + 
\frac{\sigma_\varepsilon^2}{n_s}
\left(
\mathbb{E} \big[\hat{\mathbb{E}}_s[\hat w(X)^2]\big]-1
\right).~
\@
\end{proposition}
\vspace{0.25em}
The proof of this result is in Appendix~\ref{app:subsec_proof_balance}. 
Both terms in~\eqref{eq:aidw_vs_exact_balance_gap}  are non-negative, because exact balance on the
intercept gives $\hat{\mathbb{E}}_s[\hat w(X)]=1$ and hence
$\hat{\mathbb{E}}_s[\hat w(X)^2]\ge 1$. The gap is strictly positive
whenever $\mathrm{Var}_{P_s}\{Q_0(X)\}>0$ or the balancing weights are
non-constant with positive probability.  
In a nutshell,  individual-sample reweighting can create an irreducible term even in this favorable benchmark.  
Intuitively, this is because reweighting individuals attempts to correct the random density ratio given by $W_j/\bar{W}$, the artifact of non-systematic ``noise'', and therefore inflates the variance. 
In contrast, AIDW treats such uncertainty by directly reducing the variance through dataset pooling. 



%% file: 04aihw.tex

\section{Augmented Inverse Hybrid Weighting}\label{sec:aihw}

While AIDW provides a solution for purely random perturbations, the assumption that all distributional change is random may be overly restrictive in practice. In our hospital example, we might expect some changes to be deterministic and predictable (such as differences in average patient age due to hospital specialization) while others remain random perturbations, such as day-to-day variation in referral patterns. Similar stories happen in replication studies where researchers may deliberately diversify the recruitment of participants~\citep{krefeld2024exposing,jin2024beyond}.
This motivates the hybrid shift model where a subset of covariates shifts systematically, while a random perturbation acts on top of it.
%
Built on the insights from Section~\ref{sec:aidw}, we develop Augmented Inverse Hybrid Weighting (AIHW), which interpolates between pooling  and standard weighting.

Under the hybrid  model, the source law is first shifted to the intermediate law \(P_{s,\cD}\) through the deterministic
density ratio \(w_\cD(X_\cD)\) 
and is then perturbed randomly. Throughout this section, we treat the set \(\cD\) as fixed and known. Practical procedures for selecting \(\cD\) and
estimating \(w_\cD\) are discussed in Section~\ref{sec:practical}. 

\subsection{The AIHW estimator} 
\label{app:subsec_aihw}

Let \(Q(X)=E_s[Y\given X]\) and
\(Q_\cD(X_\cD)=E_s[Y\given X_\cD]\). 
A natural starting point is an AIPW-like estimator which uses reweighting 
\(w_\cD(X_\cD)\) to adjust for the deterministic shift. We allow the regression augmentation to depend on the full
covariates, leading to the family
\[
\hat\theta(m)
=
\frac{1}{n_s}\sum_{i=1}^{n_s}
\hat w_\cD(X_{\cD,i})\{Y_i-m(X_i)\}
+\frac{1}{n_t}\sum_{i=1}^{n_t}m(X_i'),
\]
where the augmentation function $m$ belongs to the class 
\@\label{eq:def_MD}
\mathcal M_{\mathcal D}
:=
\left\{m:\mathcal X\to\mathbb R:
\mathbb E_s[Y-m(X)\mid X_{\mathcal D}]=0\right\} 
\@
to preserve the Neyman orthogonality and remove the first-order impact of
estimation error of \(w_{\mathcal D}\).  Two natural members of this family are $m=Q_\cD$, which
gives AIPW based on $X_\cD$ alone, and $m=Q$. 
Intuitively, once  the deterministic shift is corrected,  $\hat\theta(m)$ is asymptotically unbiased for
$\theta(W_{\bullet})$ while the random shift component contributes to variance. 
%
Later, we shall show that (Theorem~\ref{theorem:aihw}) every $m\in \cM_\cD$ yields an asymptotically unbiased estimator of $\theta(W_\bullet)$, and minimizing that variance (Corollary~\ref{prop:aihw-optimal-augmentation}) yields 
\begin{equation}\label{eq:lambda-star}
  m_n^*=\lambda_\cD^*\,Q+(1-\lambda_\cD^*)\,Q_\cD,
  \qquad
  \lambda_\cD^*(x_\cD)
  =\frac{w_\cD(x_\cD)/n_s+\delta_{\mathrm{dist}}^2}
        {w_\cD(x_\cD)/n_s+\delta_{\mathrm{dist}}^2+1/n_t}.
\end{equation}
The pooling weight $\lambda_\cD^*$ can be viewed as the AIDW target weight
$1-\alpha^*$ from~\eqref{eq:opt_alpha_aidw}, computed at
each $x_\cD$ with $n_s$ replaced by the local effective
source sample size $n_s/w_\cD(x_\cD)$. 

This motivates our estimator. Let $\hat Q$, $\hat Q_\cD$ and $\hat w_\cD$ be estimates from auxiliary
samples, let $\hat\lambda_\cD$ be obtained by plugging $\hat w_\cD$ and
$\hat\delta_{\mathrm{dist}}^2$ into~\eqref{eq:lambda-star}, and let
$\hat m_n:=\hat\lambda_\cD\hat Q+(1-\hat\lambda_\cD)\hat Q_\cD$. Our 
augmented inverse hybrid weighting (AIHW) estimator is then defined as 
\begin{equation}\label{eq:aihw}
  \hat\theta_{\mathrm{AIHW}}
  = \frac{1}{n_s}\sum_{i=1}^{n_s}\hat w_\cD(X_{\cD,i})\{Y_i-\hat m_n(X_i)\}
  + \frac{1}{n_t}\sum_{i=1}^{n_t}\hat m_n(X_i') .
\end{equation}

We analyze the AIHW estimator in the next subsection, including the asymptotic normality of $\hat\theta_{\text{AIHW}}^*$ subject to nuisance estimation errors, and justifications for $m_n^*$. 

\noindent\textbf{AIHW interpolates between existing approaches.}
The AIHW estimator can be viewed as interpolating between existing approaches as special cases. 
  First, if \(w_{\mathcal D}\) is close to one, then
  \(\lambda_{\mathcal D}^*\) is nearly constant, so AIHW approximately uses
  dataset-level pooling as in AIDW. In particular, when
  \(\mathcal D=\emptyset\) so $w_{\cD}\equiv 1$, this coincides with the optimal AIDW estimator~\eqref{eq:opt_alpha_aidw}.
  Second, if the distributional uncertainty \(\delta_{\mathrm{dist}}^2\) dominates
  \(w_{\mathcal D}/n_s+1/n_t\), then
  \(\lambda_{\mathcal D}^*\approx1\) and \(m_n^*\approx Q\). Thus AIHW
  approaches the reduced set covariate AIPW estimator with outcome regression \(Q(X)\).
  Third, if \(1/n_t\) dominates
  \(w_{\mathcal D}/n_s+\delta_{\mathrm{dist}}^2\), which may occur when the
  target sample is small, then \(\lambda_{\mathcal D}^*\approx0\) and
  \(m_n^*\approx Q_{\mathcal D}\); AIHW therefore approaches the AIPW
  estimator whose outcome regression uses only \(X_{\mathcal D}\).

\subsection{Asymptotic normality and optimality}
\label{subsec:aihw_justify}

%
We first show the asyptotic Gaussianity of the AIHW estimator with a generic estimated augmentation. 
The proof of Theorem~\ref{theorem:aihw} is in Appendix~\ref{app:aihw-gaussianity}. 

\vspace{-0.5em}
\begin{theorem}[Asymptotic Gaussianity]
\label{theorem:aihw}
Suppose Assumption~\ref{ass:model} holds, and
\(\mathbb E_s[Y^4]<\infty\). Fix a bounded
\(m_\infty\in\mathcal M_{\mathcal D}\). Suppose \(\hat m_n\) is
constructed from auxiliary data independent of the current samples and
\(W_\bullet\), while \(\hat w_{\mathcal D}\) is independent of the main
source sample but may depend on target data and \(W_\bullet\). Suppose there exists some constant $M>0$ such that \(\sup_x|\hat m_n(x)|\leq M\) and
\(\sup_x|\hat w_{\mathcal D}(x_\cD)|\leq M\) with probability tending to one. 
Assume
\( 
\|\hat m_n-m_\infty\|_{L^2(P_s)}
+\|\hat m_n-m_\infty\|_{L^2(P_{s,\cD})}
+\|\hat w_{\mathcal D}-w_{\mathcal D}\|_{L^2(P_s)}=o_P(1).
\)
In addition, suppose
$\|\hat w_{\mathcal D}-w_{\mathcal D}\|_{L^2(P_s)}
\|\hat m_n-m_\infty\|_{L^2(P_s)}
=o_P(n_s^{-1/2})$. For any $m\in \cM_\cD$, let 
\begin{equation}\label{eq:Vn}
  V_n(m):=\frac{1}{n_s}\mathrm{Var}_{P_s}\{w_\cD(Y-m)\}
  +\frac{1}{n_t}\mathrm{Var}_{P_{s,\cD}}(m)
  +\delta_{\mathrm{dist}}^2\,\mathrm{Var}_{P_{s,\cD}}(Y-m).
\end{equation}
If \(\liminf_{n\to\infty}JV_n(m_\infty)>0\), then
$V_n(m_\infty)^{-1/2}(\hat\theta_{\textnormal{AIHW}}-\theta(W_\bullet))
\xrightarrow{d}\mathcal N(0,1)$,
where $\hat\theta_{\mathrm{AIHW}}$ takes the form~\eqref{eq:aihw} with
$\hat m_n$ satisfying the conditions
above, and  
the randomness is over the sampling process in Assumption~\ref{ass:model} and that to construct \(\hat m_n\) and \(\hat w_{\mathcal D}\).
\end{theorem}
\vspace{-0.5em}

In summary, Theorem~\ref{theorem:aihw} shows that, as long as $\hat{m}_n$ converges to a limiting function $m_\infty\in \cM_\cD$ and the product of estimation errors in $\hat{m}_n$ and $\hat{w}_\cD$ decays faster than $n_s^{-1/2}$, the AIHW estimator is asymptotically normal with estimable variance. 
In~\eqref{eq:Vn}, the first and third terms are from source sampling and the perturbation,
and are minimized by $m=Q$; the second, from target sampling, is minimized by $m=Q_\cD$. The optimal augmentation thus trades off between the two sources. 

Corollary~\ref{prop:aihw-optimal-augmentation} characterize the optimal $m_\infty$ that minimizes the asymptotic variance, yielding the proposed AIHW estimator~\eqref{eq:aihw}. 
The proof is in Appendix~\ref{app:mse-optimal-aihw}.
\begin{corollary}
\label{prop:aihw-optimal-augmentation}
Assume \(\mathbb E_s[Y^2]<\infty\). Over
\(m\in\mathcal M_{\mathcal D}\) for which \(V_n(m)<\infty\), a
minimizer of \(V_n\) in \(L^2(P_{s,\cD})\) is $m_n^*(x)$ defined in~\eqref{eq:lambda-star}, 
which is unique wherever
\(\mathbb E_{s,\mathcal{D}}[(Q-Q_{\mathcal D})^2\mid X_{\mathcal D}]>0\).
\end{corollary}
\vspace{-0.5em}

Although $\lambda^*_{\mathcal D}$ depends on $n$, in the regime of Assumption~\ref{ass:model}, it converges to the fixed limit $\lambda_\infty$ obtained by replacing $1/n_s$, $1/n_t$, and $\delta_{\mathrm{dist}}^2$ with $1/\rho_s$, $1/\rho_t$, and $\Var(W_1)$, respectively. The corresponding $m_\infty=\lambda_\infty Q+(1-\lambda_\infty)Q_{\mathcal D}$ again lies in $\mathcal M_{\mathcal D}$, so Theorem~\ref{theorem:aihw} applies as well; the same argument also applies to $\alpha^*$ for AIDW.


\noindent\textbf{Scope of the theory.} 
The conditions in Theorem~\ref{theorem:aihw} allows $\hat{w}_\cD$ to depend on the target data and the realized perturbation; in practice, this can be achieved by cross-fitting with additional bookkeeping. 
Theorem~\ref{theorem:aihw} imposes high-level conditions on
$\hat m_n$ which, in our plug-in implementation~\eqref{eq:aihw},  depends on $\hat Q$, $\hat Q_{\cD}$
and $\hat\lambda_{\cD}$. Here, $\hat{Q}$ and $\hat{Q}_{\cD}$ (if $\cD$ is known) can again be estimated using the source samples, independently of $W_\bullet$, via cross-fitting. 
Practical use, however, may deviate from this requirement in two ways. First, the function $\hat\lambda_\cD$ often needs to be estimated in a way that depends on the realized perturbation. 
In the simulations of
Section~\ref{subsec:simu_aihw}, we find that replacing
$\hat\lambda_{\cD}$ by the oracle $\lambda_{\cD}^*$
does not visibly change the RMSE, and using $\hat\lambda_{\cD}$ achieves satisfactory coverage. 
We leave a formal
theory for such perturbation-dependent tuning 
for future work. 
Second, the theory treats $\cD$ as fixed. 
In some applications, $\cD$ can be known 
from domain knowledge, for instance when the investigators control the recruitment on certain covariates; otherwise, it can be selected with the
diagnostics of Section~\ref{section:diagnostics}. In general, guarantees that account for this selection step would require post-selection arguments, which we do not pursue in this work. Empirically, we find that AIHW performs similarly with a selection $\cD$ in simulation (Section~\ref{sec:simu}), and the real-data results are largely insensitive to the selection rule (Section~\ref{sec:experiments}).

%% file: 05implement.tex
\section{Practical implementation}\label{sec:practical}

In this section, we discuss practical considerations in implementing the AIDW and AIHW estimators, including (i) estimating the distributional distance \(\delta_{\mathrm{dist}}^2\), which affects the pooling strategy; (ii) identifying the subset \(\mathcal{D}\) that exhibits deterministic shifts; and (iii) estimating the deterministic weight function \(w_{\mathcal D}(x_\cD)\) in AIHW.  
These procedures
build on the theoretical results above, 
but may involve data-driven
choices not covered there. The challenges of analyzing such data-driven choices are analogous to those of post-selection arguments, which we leave open; as such, careful readers may treat these tools as heuristic, and we evaluate the full pipeline empirically
in Sections~\ref{sec:experiments}
and~\ref{sec:simu}.
%

\subsection{Estimating the distributional distance}\label{sec:estimation-of-delta}

Building on \citet{jeong2024out}, we use a plug-in calibration rule for $\delta_{\text{dist}}^2$ motivated by  the distributional CLT similar to~\eqref{eq:dclt}. 
We introduce the method for the hybrid model, with the purely random shift model a special case  with $\cD=\varnothing$.


Consider test functions $\phi_1,\ldots,\phi_L$ satisfying
$\mathbb{E}_s[\phi_\ell]=0$ and
$\mathbb{E}_s[w_{\mathcal D}(X_{\mathcal D})\phi_\ell]=0$. 
One convenient choice is the residualized form 
$\phi_\ell(X)=r_\ell(X)-\mathbb{E}_s[r_\ell(X)\mid X_{\mathcal D}]$. 
Corollary~\ref{cor:hybrid_phi_discrepancy} in Appendix~\ref{subsec:distributional-clt-matched} suggests that $s_{\phi,n}^{-1}\hat\EE_{s}[\phi(X)] - \hat\EE_{t}[\phi(X)]\stackrel{d}{\to} \cN(0,1)$, where the leading-order variance is 
$
s_{\phi,n} 
\approx
\frac{1}{n_s}\Var_{P_s}(\phi_\ell)
+
 (\frac{1}{n_t}+\delta_{\mathrm{dist}}^2 )\Var_{P_{s,\mathcal D}}(\phi_\ell)$. 
For a moderately large collection of approximately \emph{orthogonal} test
functions $\phi_1,\ldots,\phi_L$, the standardized source-target
discrepancies are approximately uncorrelated and jointly Gaussian. This suggests that
\@\label{eq:calib_idea}
\frac{1}{L}\sum_{\ell=1}^L
\frac{
 (\hat{\mathbb E}_s[\phi_\ell]-\hat{\mathbb E}_t[\phi_\ell] )^2
}{
\frac{1}{n_s}\operatorname{Var}_{P_s}(\phi_\ell)
+
(
\frac{1}{n_t}+\delta_{\mathrm{dist}}^2)
\operatorname{Var}_{P_{s,D}}(\phi_\ell)
}
\approx 1,
\@
since each standardized squared discrepancy is approximately $\chi^2_1$.
This motivates estimating $\hat\delta_{\mathrm{dist}}^2$ by solving this moment equation, set as the nonnegative solution to
\[
\frac{1}{L}
\sum_{\ell=1}^L
\frac{
\bigl(\hat{\mathbb E}_s[\phi_\ell]-\hat{\mathbb E}_t[\phi_\ell]\bigr)^2
}{
(1/n_t+\delta^2)\hat v_{T,\ell} + (1/n_s)\hat v_{S,\ell}
}
= 1,
\]
where $\hat v_{S,\ell}$ and $\hat v_{T,\ell}$ are the source and target
sample variances of $\phi_\ell$. 
The above display is a non-increasing function of $\delta^2$, so we can compute the solution by bisection search; if the left-hand side is below one at \(\delta=0\), we set \(\hat\delta_{\mathrm{dist}}^2=0\). 
Correlation among the test functions
does not affect the mean of the left-hand side
of~\eqref{eq:calib_idea}, but it reduces the precision of
$\hat\delta_{\mathrm{dist}}^2$, so we use a moderately large collection
of whitened functions in the experiments.

\subsection{Diagnosing the hybrid shift model}\label{section:diagnostics}

Another element in AIHW is the subset $\cD$ which captures the deterministic shift component. It needs to be estimated from data when not   available from domain knowledge. 
In the following, we discuss (1) how to diagnose whether a chosen set of features capture the deterministic shift, which further inspires (2) how to use such diagnosis to screen for $\cD$.  

Intuitively, it is possible distinguish $\cD$ from others using data because the two
components of the shift act on different scales. A
deterministic shift moves the mean of an affected
covariate by a fixed amount, whereas the random
perturbation and sampling move it only by
$O\{(n_s^{-1}+n_t^{-1}+\delta_{\mathrm{dist}}^2)^{1/2}\}
=O(J^{-1/2})$ (Theorem~\ref{thm:matched-clt}). Hence, once the dependence on $X_\cD$ is removed, the
standardized source--target discrepancies are 
approximately $N(0,1)$ for $k\notin\cD$, but not for those in $\cD$. A covariate can then be flagged when its
discrepancy is larger than the random perturbation suggests. 
We now formalize these ideas.

\noindent\textbf{Model diagnostic.} For diagnostic, we use the fact that any feature $X_k$ for $k\notin \cD$ exhibit random-perturbation-like behavior after reweighting by $\cD$.  
For a candidate covariate $X_k$ and regression function estimate $\hat{Q}_k(X_\mathcal{D}) = \hat{\EE}_s[X_k \given  X_\mathcal{D}]$, the following corollary applies the general AIHW result to the residual contrast, whose proof is in Appendix~\ref{app:proof-model-check}.

\vspace{-0.5em}
\begin{corollary}[AIHW residual diagnostic]\label{cor:model-check}
Fix \(k\notin\mathcal D\), and define 
$Q_k(x_{\mathcal D})
:=
\mathbb E_s[X_k\mid X_{\mathcal D}=x_{\mathcal D}]$, 
$\psi_k(X):=X_k-Q_k(X_{\mathcal D})$. 
Let \(\hat Q_k\) be an auxiliary-sample estimate of \(Q_k\), and 
\(\hat\psi_k(X):=X_k-\hat Q_k(X_{\mathcal D})\). Suppose the conditions of
Theorem~\ref{theorem:aihw} hold for pseudo-outcome
\(\widetilde Y\equiv0\), oracle augmentation \(m_\infty=-\psi_k\), and
fitted augmentation \(\hat m_n=-\hat\psi_k\), where the moment condition is replaced by 
\(\mathbb E_s[X_k^4]<\infty\) and \(\psi_k\) be bounded. Define
$\hat D_k
:=
\hat{\mathbb E}_s [\hat w_{\mathcal D}(X_{\mathcal D})\hat\psi_k(X) ]
-\hat{\mathbb E}_t[\hat\psi_k(X)]$
and
$s_k^2
:=
\frac{1}{n_s}\Var_{P_s}
\big\{w_{\mathcal D}(X_{\mathcal D})\psi_k(X)\big\}
+(\frac{1}{n_t}+\delta_{\mathrm{dist}}^2)
\Var_{P_{s,\mathcal D}}\{\psi_k(X)\}$.
If \(\liminf_{n\to\infty}J s_k^2>0\), then
\begin{equation}\label{eq:diag_stat}
 \hat D_k / s_k\xrightarrow{d}\mathcal N(0,1).
\end{equation}
\end{corollary}
\vspace{-0.5em}

Corollary~\ref{cor:model-check} inherits the same fixed-\(\mathcal D\), sample-splitting, and nuisance-rate conditions from Theorem~\ref{theorem:aihw}.
In practice we may use plug-in estimates for \(s_k^2\) and \(\delta_{\mathrm{dist}}^2\). 
We define a concrete plug-in estimator $\hat{s}_k^2$ in Appendix~\ref{app:subsec_detail_diagnostic}. 

Corollary~\ref{cor:model-check} inspires the following diagnostic of hybrid shift model with a given set $\cD$. 
Let $M = d-|\cD|$. When $M>0$, the diagnostic flags the hybrid shift model for subset $\mathcal{D}$ if  
\begin{equation}\label{eq:tstat_diag}
     \max_{k\notin \cD} |\hat{\Delta}_k|/\hat{s}_k> \Phi^{-1}(1 - 0.025/M),
\end{equation}
where $\Phi^{-1}$ denotes the standard Gaussian quantile function. 
%
A large value of the standardized contrast on the left suggests that the current working model with $\cD$ misses deterministic structure involving some $X_k$ or covariates correlated with it. 

\noindent\textbf{Screening for $\cD$.} The same idea suggest screening variables and construct an estimate for $\cD$. First, the rule discussed above can be applied iteratively to build a candidate deterministic covariate set: one may include the indices with the largest standardized contrast on the left-handed side of~\eqref{eq:tstat_diag} until those outside of $\cD$ do not exceed a threshold. 
Second,~\eqref{eq:tstat_diag} suggest that the standardized contrast should be approximately normal for $k \not \in \cD$ and take larger values for those $k \in \cD$, which inspires a two-group structure for these statistics. From a heuristic perspective, one may use clustering methods such as a Gaussian mixture model to identify the two groups of variables. In the real data analyses in Section~\ref{sec:experiments}, the AIHW estimator is implemented with these two ideas, and we observe robust performance with respect to the screening procedure.

\subsection{Estimating the deterministic weights}

Finally, the AIHW estimator involves a weight function $w_{\cD}$ that only depends on $X_\cD$, which needs to be estimated unless sufficient domain knowledge is available. 
The subtlety here is that $w_\cD$ is not the realized density ratio between the source/target covariate laws. Instead, it is the deterministic component of that density ratio. Thus, a standard density-ratio estimator trained to capture all
source and target differences may fit random fluctuations that does not need to be reweighted away. 

The key theoretical observation is that under the hybrid model, the realized density ratio of $X_\cD$ is $w_\cD(X_\cD)$ multiplied by $\EE_s[W_{j(X,Y)}/\bar W \given X_\cD]$, where
$j(X,Y)$ indexes the cell containing $(X,Y)$, so the
random factor averages the weights of the many cells
consistent with $X_\cD$. This average is close to one:
for any fixed function $g$ of $X_\cD$, the random factor
changes $\mathbb{E}[g(X_\cD)]$ by only $O_P(J^{-1/2})$
(Theorem~\ref{thm:matched-clt}), but a flexibility density ratio estimator may average over fewer cells and partially chase the individual weights. 

Our implementation follows this logic  to estimate $w_\cD$ as a reduced or regularized
density-ratio weight on \(X_{\mathcal D}\). 
We first estimate the full-$X$ density ratio, and project the logarithm of the weights onto  
\(X_{\mathcal D}\) to construct the projected weights. Working on the log scale is natural here because~\eqref{eq:hybrid} makes the realized log-likelihood ratio a sum of $\log w_{\mathcal D}(x_{\mathcal D})$ and $\log(W_j/\bar W)$, and the latter is largely averaged out by this projection. 
Other possible implementations include
logistic classifiers and kernel mean matching with proper regularization, whose theoretical properties are beyond the scope of this work. 

Finally, before using the fitted weights in AIHW, we source-normalize them so that 
$\hat{\mathbb E}_s[\hat w_{\mathcal D}(X_{\mathcal D})]=1$. 
We also recommend estimating the weights on an auxiliary fold or using
cross-fitting, and applying mild clipping when the
estimated weights are unstable.

%% file: 06real.tex
\section{Real data experiments}\label{sec:experiments}


We demonstrate the efficacy of the proposed methods in generalizing statistical parameter estimation in three real-world multi-site datasets. 
We take pairs of sites to emulate a generalization task, and use the ``target'' site full-data estimator as the oracle benchmark. 

We chose the datasets to span diverse distribution shift scenarios: a pure random shift model is plausible in the Pipeline data (Section~\ref{subsec:real_pipeline}), while in the KSJ data (Section~\ref{subsec:real_ksj}) a hybrid shift model is sensible due to intended diversification of sites by the investigators; finally, the ACS income data (Section~\ref{subsec:real_acs}) examines shifts across U.S.~states which are the most difficult to model. 
Our methods demonstrate  robust performance even when the dataset might not be best described by the shift model they are designed for.

\subsection{Evaluation pipeline}

We use the same evaluation framework throughout. 
Each dataset consists of per-site data $\mathsf{D}^{(k)} = \{D_i^{(k)}\}_{i=1}^{n_k}$ for site $k=1,\dots,K$, where each $D_i^{(k)}$ is an individual observation. 
We assume within-site data are i.i.d., $D_i^{(k)}\sim P^{(k)}$, and $P^{(k)}$'s differ across sites.   
The parameter of interest is $\theta_k = \theta(P^{(k)})$ for a functional $\theta(\cdot)$. The full data in a site yields an unbiased estimator $\hat\theta_k = \theta(\mathsf{D}^{(k)})$ used for evaluation.

The Pipeline data in Section~\ref{subsec:real_pipeline} and the KSJ data in Section~\ref{subsec:real_ksj} are from randomized experiments, $D_i=(X_i,Y_i,T_i)$, where $X_i$ is the background characteristics for the participants, $T_i\in \{0,1\}$ are randomly assigned treatments, and $Y_i\in \RR$ is the observed outcome. 
For the ACS-income data in Section~\ref{subsec:real_acs}, $D_i=(X_i,Y_i)$ for features $X_i$ and outcomes $Y_i$. 
For randomized experiments, the parameter is the average treatment effect $\EE^{(k)}[Y\given T=1]-\EE^{(k)}[Y\given T=0]$, and we apply our framework to the two group means separately. 
For others, the parameter of interest is the mean outcome $\EE^{(k)}[Y]$.

We take each pair $(i,j)$ for $i\neq j$, $1\leq i,j\leq K$. The site $i$ is taken as the source site, and the methods have access to the full data $\mathsf{D}^{(i)}$; the site $j$ is treated as the target site, and the methods have access to the covariates $\cX^{(j)}:=\{X_\ell^{(j)}\}_{\ell=1}^{n_j}$ only. 
The methods compute an estimator $\hat\theta_{i\to j}= f(\mathsf{D}^{(i)},\cX^{(j)})$ for $\theta_j$, around which they construct a predictive interval $\hat{C}_{i\to j}$ for $\hat\theta_j$. We use the empirical estimator $\hat\theta_j$ as a benchmark to evaluate the performance of $f(\mathsf{D}^{(i)},\cX^{(j)})$. We compute the root mean-squared error $\{\frac{1}{K(K-1)}\sum_{i\neq j} (\hat\theta_{j} - \hat\theta_{i\to j})^2\}^{1/2}$ to assess the accuracy of the estimator, 
and coverage $\frac{1}{K(K-1)}\sum_{i\neq j} \ind\{ \hat\theta_j \in \hat{C}_{i\to j}\}$ to assess the reliability of uncertainty quantification. 
The methods under comparison include:
    (i) \texttt{AIPW}: the AIPW estimator~\citep{robins1994estimation} with two-fold cross-fitting~\citep{chernozhukov2018double}, which combines outcome regression and covariate shift adjustment. 
    (ii) \texttt{SBW}: the stable-balancing weights estimator~\citep{zubizarreta2015stable}, which is the reweighted estimator with minimal-variance weights that balance the feature means in source and target sites. 
    (iii) \texttt{AIDW} with the parameter $\alpha$ chosen following~\eqref{eq:opt_alpha_aidw}, and we use two-fold cross-fitting to fit the regression functions. 
    (iv-v) \texttt{AIHW} assuming a hybrid distribution shift. We use two-fold cross-fitting, where the covariate shift subset is selected by either \texttt{gaussian-mix} or \texttt{t-stat} in the same fold of data used to fit the regression models. 
For each method, we use the associated asymptotic variance to compute the Wald prediction intervals. See Appendix~\ref{subsec:implementation} for the detais on implementation.

The AIPW and SBW estimators are designed for pure covariate shift settings. Under the covariate shift assumption, together with suitable regularity conditions, both the AIPW and the SBW estimators are consistent and asymptotically normal~\citep{dahabreh2020extending,chernozhukov2018double,wang2020minimal}.

\subsection{The Pipeline project: generalizing across replication sites}
\label{subsec:real_pipeline}

The first case study concerns the datasets from the Pipeline project~\citep{schweinsberg2016pipeline}. It is a multi-site replication study where 25 laboratories across the world (contributing 29 populations) independently replicate experiments to test 10 scientific hypotheses. The participating sites are invited by the project lead because in these sites ``the original finding was theoretically expected to replicate using the original materials''. 

\noindent\textbf{Plausible random shift.} 
The invitation criterion makes systematic differences between sites
unlikely, while many small deviations in implementation are
inevitable. 
Analysis of this dataset in~\citet{jin2024beyond} found evidence that supported the purely random-perturbation model across sites. 
Under this model, AIHW should behave similar to AIDW, and ignoring the distributional uncertainty should lead to under-coverage. 
Of course, however, it is impossible to know the true shift type.

\begin{figure} 
    \centering
    \includegraphics[width=\linewidth]{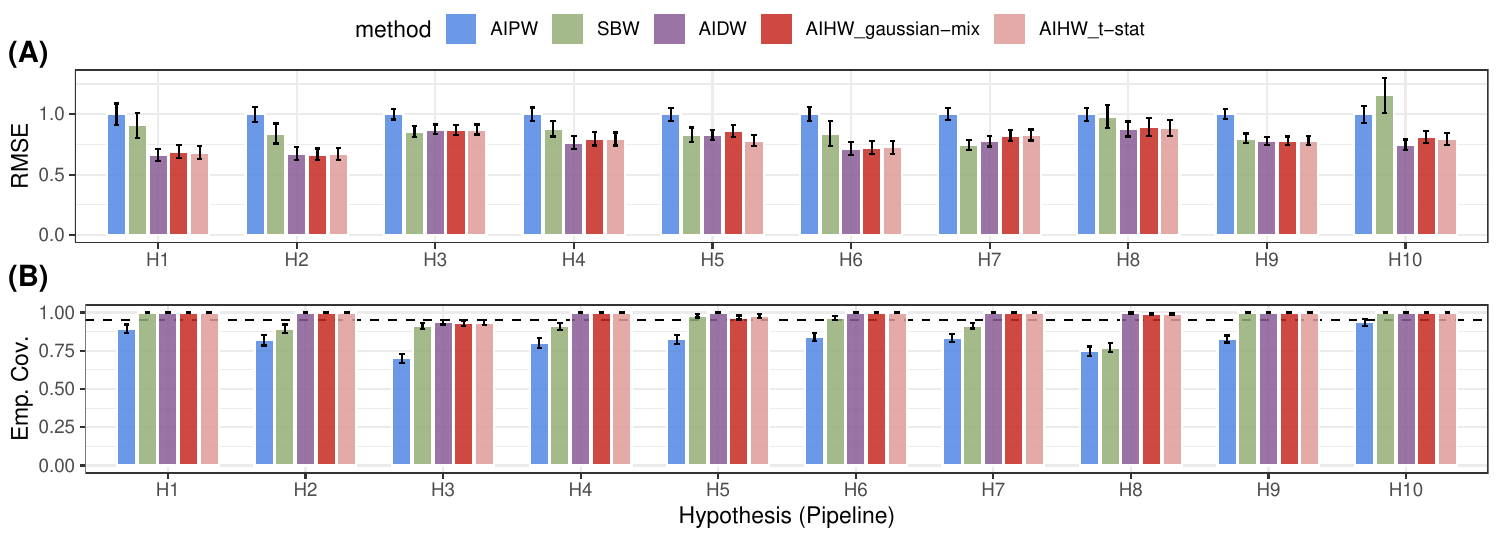}
    \caption{Empirical (A) RMSE and (B) coverage averaged across all site pairs for each hypothesis in the Pipeline dataset. The black bars show $\pm 1.96\times \text{std}$, and the dashed line in (B) is the nominal $95\%$ level.}
    \label{fig:real_pipeline_rmse_cov}
\end{figure}

\noindent\textbf{Results.}
Figure~\ref{fig:real_pipeline_rmse_cov} presents the RMSE (A) and empirical coverage (B) between site pairs for testing each hypothesis, where the RMSE is normalized by the AIPW estimator's RMSE for visualization. We observe that the covariate-shift-based estimators (AIPW and SBW) lead to large estimation error and low coverage. This might be due to (i) violation of the covariate shift assumption, and (ii) unstable weight estimation which inflates the variance and contributes to the large RMSE. 
Even though the SBW estimator explicitly seeks small-variance weights, it can still lead to large RMSE and low coverage.

In contrast, AIDW and AIHW variants achieve both low RMSE and high coverage. For Hypothesis H1, the reduction in RMSE (which includes the irreducible error) by AIDW relative to AIPW is up to 40\%. 
The three of our methods are comparable in most cases, though sometimes AIDW can be slightly more accurate. The variable selection methods did not make a huge difference. The prediction intervals, which account for distributional uncertainty, provide reliable coverage. 
The superior performance of AIDW appears consistent with our argument for the random-perturbation model. Meanwhile,  AIHW  remains robust under weak systematic shift.



\subsection{The KSJ data: generalizing across diversified sites}
\label{subsec:real_ksj}


The second dataset was collected by~\citet{krefeld2024exposing}, which we refer to as the KSJ data. 
The distinctive feature is that the sites are \emph{selected to differ} by the authors (following~\citet{jin2024beyond}, we take panels from
studies 1 to 2, totaling 13 panels for 4 hypotheses) in order  to examine the variability of causal effects.

\noindent\textbf{Plausible hybrid shift.} The deliberate diversification suggests a deterministic shift in some
covariates, yet one may still expect random perturbations due to inevitable deviations in the process, suggesting a hybrid model. If so, one would expect the discrepancies to resemble random perturbations once the deterministic component is removed.

\noindent\textbf{Results.} The RMSE and coverage averaged over site pairs grouped by the same target population are shown in Figure~\ref{fig:real_ksj_rmse_cov}. AIPW and SBW tend to have higher RMSE than AIDW and AIHW, and for some target panels the reduction in RMSE (which includes the irreducible error) is substantial. 
While AIPW often undercovers, the coverage of SBW is close to the target level.
The coverage of AIDW and AIHW is at or above nominal for all  panels. 
Supporting the hybrid model argument, the AIHW variants, especially when using the Gaussian mixture variable selection, typically achieve the lowest RMSE. 
Surprisingly, the AIDW method also achieves quite similar RMSE, suggesting that the random-shift model can be a useful working approximation in diverse practical scenarios.

\begin{figure}[htbp]
    \centering
    \includegraphics[width=\linewidth]{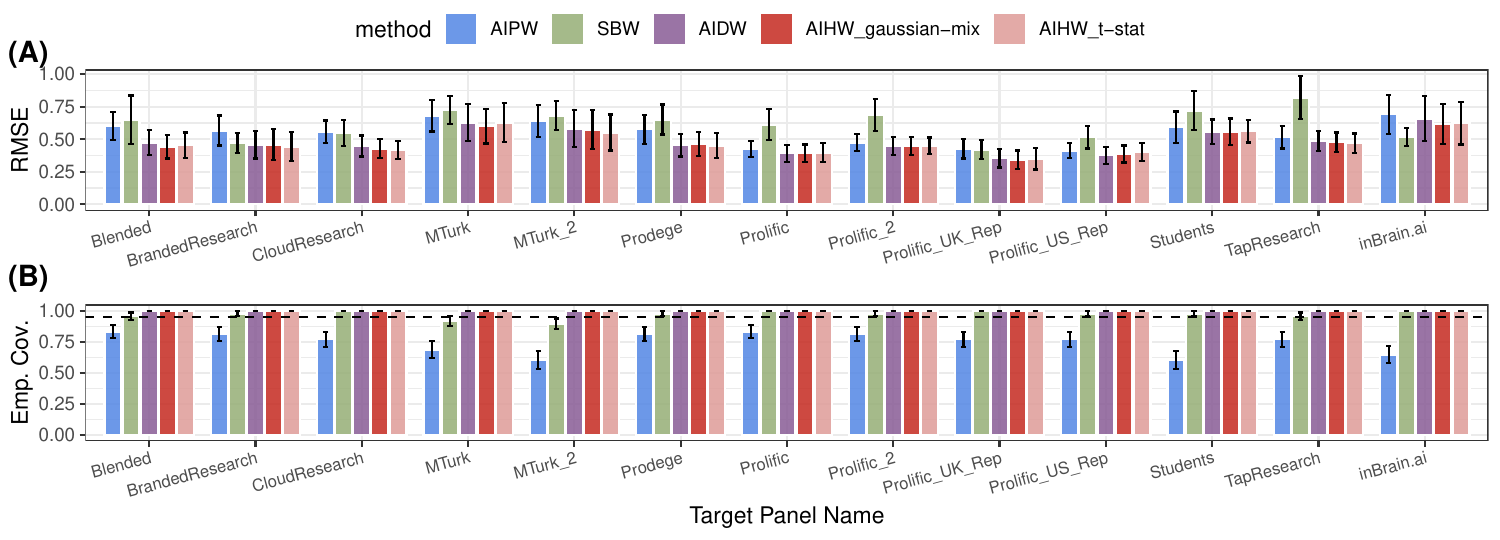}
    \caption{Empirical (A) RMSE and (B) coverage averaged across site pairs with a specific target population in the KSJ dataset. The black bars show $\pm 1.96\times \text{std}$, and the dashed line in (B) is the nominal $95\%$ level.}
    \label{fig:real_ksj_rmse_cov}
\end{figure}

We further perform the model diagnosis outlined in Section~\ref{section:diagnostics} to examine whether the hybrid model is sensible and whether the correction for the deterministic shift is effective. 
Figure~\ref{fig:real_ksj_diagnostic} presents the QQ-plots for the mean differences in $X_{-\hat{\cD}}$ before and after correcting for the deterministic shift, where $\hat\cD$ is selected by the \texttt{gaussian-mix} method. 
We present five representative pairs; the other pairs or correcting with \texttt{t-stat}  yield similar patterns. By Corollary~\ref{cor:model-check}, when  the deterministic component is (approximately) correctly accounted for, the residual covariate mean differences~\eqref{eq:diag_stat} is approximately normal. The second row of Figure~\ref{fig:real_ksj_diagnostic} supports this phenomenon and  justifies our uncertainty quantification method. 

\begin{figure}[htbp]
    \centering
    \includegraphics[width=\linewidth]{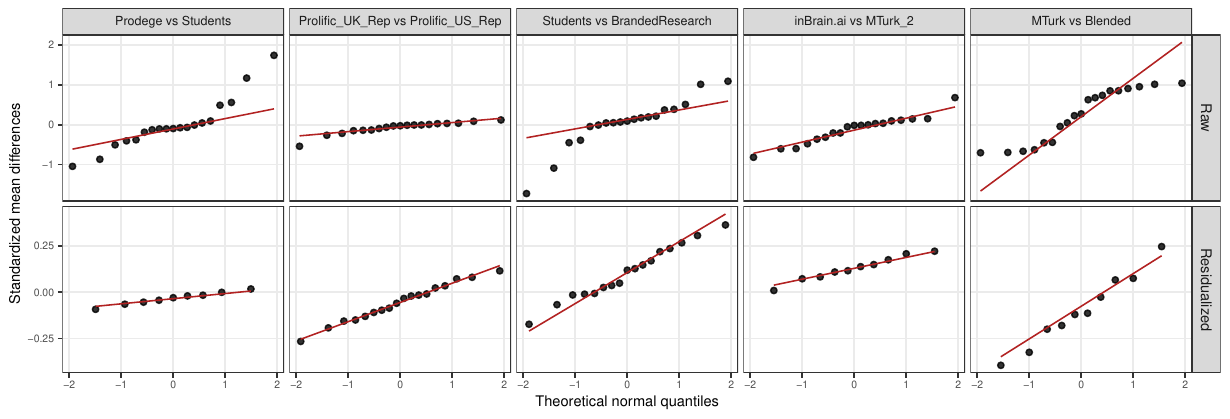}
    \caption{QQ-plot of covariate mean differences across five representative pairs. The first row shows the difference in covariate mean for all covariates. The second row shows the residual~\eqref{eq:diag_stat} after correcting for the deterministic shift.}
    \label{fig:real_ksj_diagnostic}
\end{figure}

\subsection{ACS income data: generalizing across states}
\label{subsec:real_acs}

The third dataset is the ACS income data derived from the United States census data, which we process based on the pipelines in~\cite{liu2023need,ding2021retiring}.  
In this dataset, each site is a state.
The response variable $Y_i\in \{0,1\}$ indicates whether the individual's income is above 50,000 USD. The parameter of interest is the mean response.

\noindent\textbf{In-the-wild shift?} For  the geographical sites, the distribution shift   can be challenging to model. 
Arguably, it is unclear a priori which model may fit. 
This dataset thus offers a stress test in scenarios where any model can be misspecified.

\noindent\textbf{Results.} The RMSE and coverage for target estimators averaged across pairs with the same target state are  in Figure~\ref{fig:real_acs_rmse_cov}. We observe consistent improvement of AIDW and AIHW upon AIPW and SBW. The improvement is especially substantial for target states like MO, NC, and TN where these baselines suffer from large RMSE, showing the robust performance under challenging distribution shifts. AIHW performs slightly better than  AIDW for most target states, while the impact of variable selection method remains small. Figure~\ref{fig:real_acs_by_distance} in Appendix~\ref{app:subsec_add_exp} further shows that the improvement is substantial for state pairs that differ most in the covariate means.

Reliable quantification of uncertainty seems particularly challenging in this dataset. The state samples are large, so sampling variance is small and intervals that ignore distributional uncertainty suffers from low coverage. The intervals of AIDW and AIHW  substantially improves upon the weighting methods but can still undercover for several states such as CA, NY and TN.  
Again, we emphasize that we take this dataset as a stress test, and do not claim that the distribution shift models proposed in this paper are perfect. We believe there remains room for future work to better capture the shift patterns present in such datasets. Still, regardless of whether these models hold exactly, our AIDW and AIHW estimators remain robust for effect generalization under real distribution shifts. 

\begin{figure}[htbp]
    \centering
    \includegraphics[width=\linewidth]{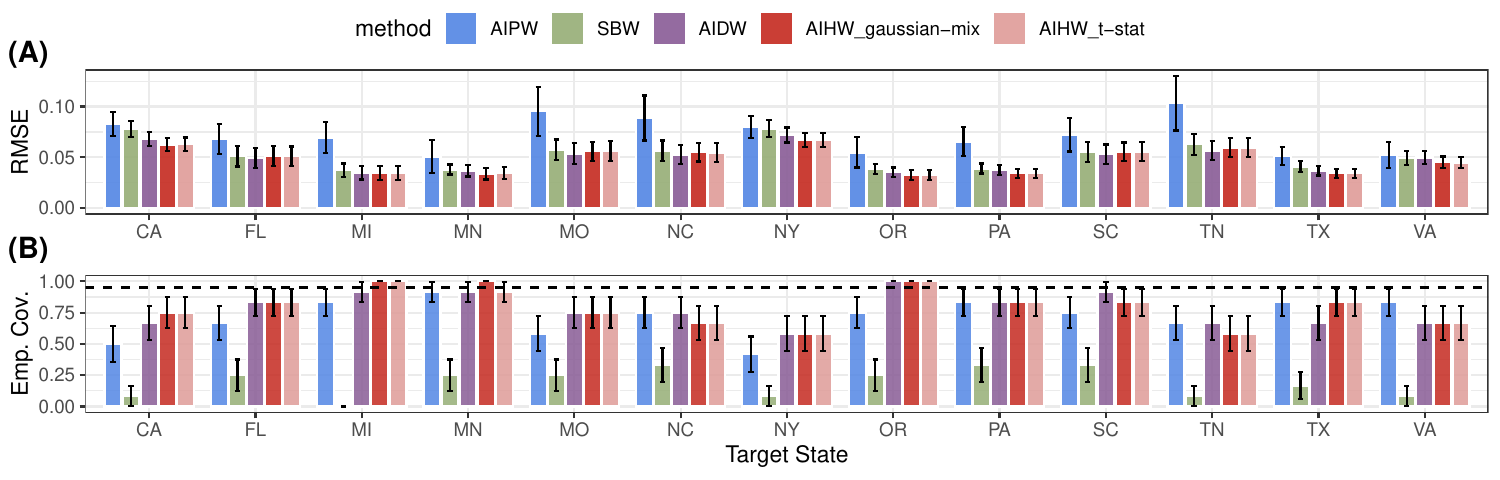}
    \caption{Empirical (A) RMSE and (B) coverage across pairs of states with a specific target state in the ACS-income data. The black bars show $\pm 1.96\times \text{std}$, and the dashed line in (B) is the nominal $95\%$ level.}
    \label{fig:real_acs_rmse_cov}
\end{figure}



%% file: 07simu.tex

\section{Simulation studies}
\label{sec:simu}
We complement the real data analyses with simulation studies to examine the robustness of our methods with estimated quantities by comparing them to (semi-)oracle variants. 

\subsection{Simulations under random perturbations}
\label{subsec:simu_aidw}

We first study the pure random shift setting to examine the robustness of AIDW with estimated $\hat\delta_{\text{dist}}^2$ and compare its performance with covariate shift adjustment methods. 
We generate source and target populations
under the random-perturbation model and vary the scaled perturbation strength $\tau := n_t \delta_{\mathrm{dist}}^2
    \in \{0,0.5,1,2,4\}$. 
We consider both linear (estimated by OLS) and nonlinear (estimated nonparametrically by random forests) true outcome models. In each case, we compare AIPW and SBW  with two versions of AIDW:  
\emph{oracle-$\delta$}  that uses the true $\delta_{\text{dist}}^2$ in
the optimal pooling coefficient, and the \emph{plug-in} version that uses an estimator $\hat\delta_{\mathrm{dist}}^2$ as in Section~\ref{sec:estimation-of-delta}. All outcome regressions are fit by source-side
cross-fitting. Further estimation and data-generating details are in Appendix~\ref{app:simulation-details}.

We evaluate the empirical RMSE and coverage over $N=1000$ independent runs in both linear/nonlinear settings in
Figure~\ref{fig:simu_aidw_linear}. 
In both settings, the oracle and plug-in AIDW estimators
behave very similarly across perturbation strengths. This indicates that the procedure in Section~\ref{sec:estimation-of-delta} works effectively with the AIDW estimator. 
AIDW is particularly advantageous when the random shift is weak to moderate;
as the perturbation becomes stronger, its optimal pooling rule places
increasing weight on the observed target covariates and the RMSE differences
among the methods decrease.  
The distinction is sharper for inference. The AIDW intervals
remain close to the nominal $95\%$ level across perturbation strengths, with
the plug-in procedure close to the oracle. In
contrast, the coverage of AIPW and SBW deteriorates as
$\tau$ increases, even when their point-estimation errors remain comparable, showing the consequences of ignoring the distributional uncertainty.

\begin{figure}
    \centering
    \includegraphics[width=\linewidth]{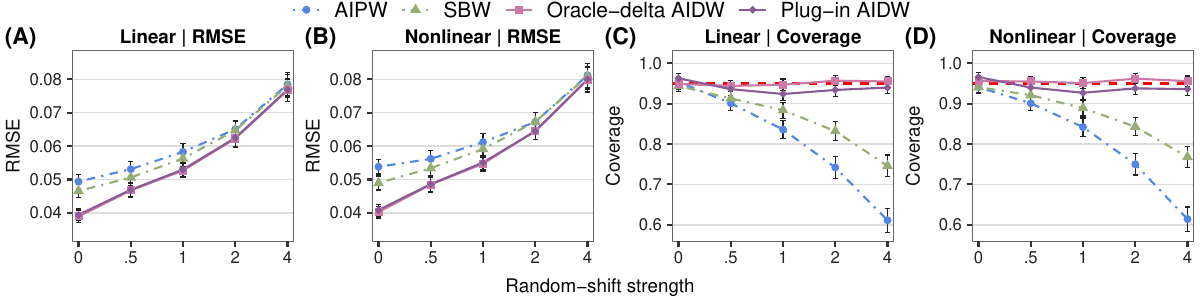}
    \caption{Empirical RMSE (A,B) and coverage (C,D) of AIPW, SBW, and AIDW with oracle $\delta_{\text{dist}}^2$ (\texttt{Oracle-delta AIDW}) and estimated $\hat\delta_{\text{dist}}^2$ (\texttt{Plug-in AIDW}) under random shift with linear and nonlinear regression functions. The mean and error bars are reported over $N=1000$ independent runs. The red dashed lines in (C,D) indicate nominal coverage level $0.95$.}
    \label{fig:simu_aidw_linear}
\end{figure}

\subsection{Simulations under hybrid shifts}
\label{subsec:simu_aihw}

We next study the hybrid model.  We begin with a linear outcome model so that the
role of the two shift components can be isolated without additional
complications from errors in outcome learning. 
The deterministic shift component $w_\cD$ gives a Gaussian location shift, whose strength is controlled by a parameter $\kappa$. 
First, we fix the scaled random-shift strength at
$\tau=n_t\delta_{\mathrm{dist}}^2=0.5$ and increase the strength $\kappa$ of
the deterministic shift.  Second, we fix $\kappa=0.75$ and vary
$\tau\in\{0,0.5,1,2,4\}$.  We compare the AIHW estimator (with the $t$-stat method in Section~\ref{section:diagnostics}) for learning $\hat{\cD}$, called \texttt{Adaptive AIHW}) with
AIDW, AIPW, SBW, and 
    \texttt{Theory AIHW}, where $\cD$,  distance $\delta_{\text{dist}}^2$, and pooling function $\lambda_\cD^*$ are known to the procedure, whereas only the functions $Q(x)$, $Q_\cD(x)$, and $w_{\cD}$ are estimated; and 
    \texttt{Known-D AIHW}, where the set $\cD$ is known to  the procedure and other quantities, including $\delta_{\text{dist}}^2$, $\lambda_\cD^*$, $Q(x)$, $Q_\cD(x)$, and $w_{\cD}$ are estimated. 
The last two variants help assess the gap
between the theoretical and fully data-adaptive implementations.  Details of
the data-generating process and the methods are  in
Appendix~\ref{app:simulation-details}.

\begin{figure}
    \centering
    \includegraphics[width=\linewidth]{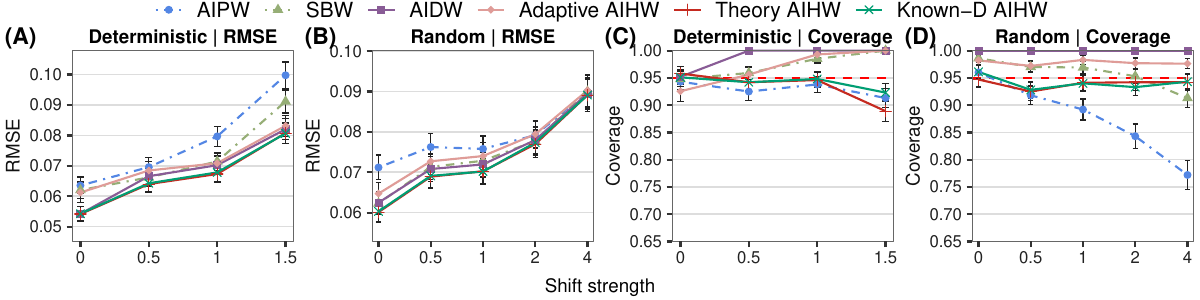}
    \caption{Empirical RMSE (A,B) and coverage (C,D)   under hybrid shifts with a linear outcome model.
Panels (A,C) vary strength $\kappa$ of the deterministic shift while fixing 
$\tau=n_t\delta_{\mathrm{dist}}^2=0.5$. Panels (B,D) vary $\tau$ while
fixing $\kappa=0.75$.  The results are reported over $N=1000$ independent runs. The red dashed lines in (C,D) indicates nominal coverage level $0.95$.}
    \label{fig:simu_aihw_core}
\end{figure}

The results are summarized in Figure~\ref{fig:simu_aihw_core}. When the deterministic
component grows (Panels A,C), estimation becomes more difficult
for all methods, but AIHW  remain among the most accurate. AIPW exhibits the largest increase in RMSE,
while SBW also degrades as the systematic discrepancy grows. 
AIDW remains competitive in point estimation, but its coverage becomes increasingly conservative; a similar
conservativeness appears for the fully adaptive AIHW estimator at large
$\kappa$ values. The Theory and Known-$\mathcal D$ variants have similar RMSE to
Adaptive AIHW, suggesting that the estimation error in $\mathcal D$ and the associated
plug-in quantities make little difference.
When the random component is increased while holding the deterministic component fixed (Panels B,D), the RMSE increases with
$\tau$ for all procedures, with AIHW  generally attaining the
smallest errors. AIPW increasingly undercovers as the random
shift increases, and SBW also undercovers when $\tau$ is large. In contrast, the Theory AIHW intervals remain close to the nominal
level while Adaptive AIHW is somewhat conservative.
Together, these results show that the nuisance-estimation procedures work effectively with the AIHW estimator.

\begin{figure}
    \centering
    \includegraphics[width=\linewidth]{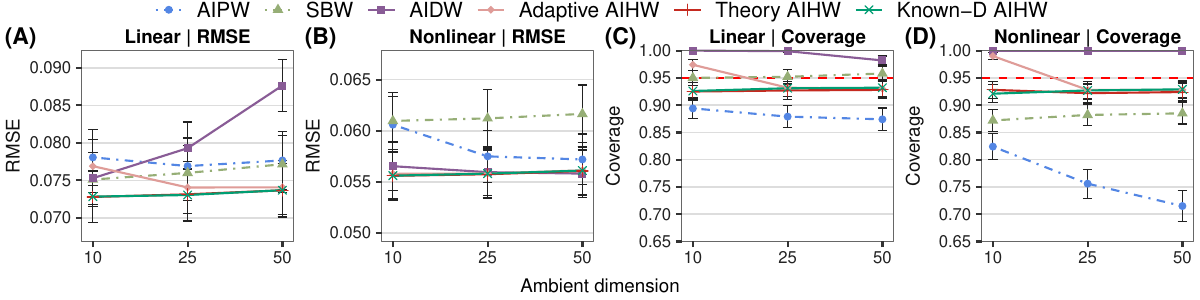}
    \caption{Empirical RMSE (A,B) and coverage (C,D) of various methods with increasing feature dimension under linear and nonlinear outcome models, fixing the distribution shift strengths.
The results are reported over $N=1000$ independent runs. The dashed lines indicate the nominal coverage level $0.95$.}
    \label{fig:simu_aihw_highd_linear}
\end{figure}

Furthermore, we fix $\tau=1$ and $\kappa=0.75$ but vary the dimension of the features $p\in \{10, 25,50\}$ in both linear and nonlinear settings. Here the regression function $Q(x)$  involves a fixed number of features, and $p$ affects the number of features over which random discrepancies are observed. Figure~\ref{fig:simu_aihw_highd_linear} summarizes the results and further
details are in Appendix~\ref{app:simulation-details}.  
AIHW remains stable as the
feature dimension increases. 
The three AIHW variants are similar across settings, and the estimation of the unknown components 
does not introduce substantial error in AIHW even with many randomly shifting covariates. 
Under the linear model, AIDW becomes
less accurate at $p=50$, and in both settings its intervals are highly
conservative, because it 
attributes the systematic shift to distributional uncertainty. 
Among
the baselines, AIPW undercovers in both settings, most severely under
the nonlinear model; SBW is competitive under the linear model but less accurate
and undercovering under the nonlinear one.
Overall, the results suggest that our methods remain effective as the number of nuisance discrepancies grows. 

%% file: 10appendix.tex

\section{Deferred discussion}

\subsection{Estimators beyond mean outcome}
\label{app:extend_linear}

In the main text we primarily focus on mean outcome estimation with full observations $(X,Y)$. Here we briefly discuss how our framework generalizes to average treatment effect estimation. We assume the standard completely randomized experiments, where one has access to the full observations  $(X_i,T_i,Y_i)_{i=1}^{n_s}$ in the source site and covariates $\{X_i'\}_{i=1}^{n_t}$ in the target site. The full observations are generated from i.i.d.~data $(X_i,Y_i(1),Y_i(0))_{i=1}^{n_s}\iid P_s$ with i.i.d.~treatment indicators $T_i\iid \text{Bern}(p)$ for some constant $p\in (0,1)$, independent of everything else. The observed outcome is then $Y_i=Y_i(T_i)$ following the SUTVA~\citep{imbens2015causal}. 
The estimand is $\EE[Y(1)-Y(0)] = \EE[Y\given T=1] - \EE[Y\given T=0]$, for which a difference-in-mean estimator is $\hat\theta_s = \hat\EE_s[Y\given T=1] - \hat\EE_s[Y\given T=0]$, where $\hat\EE_s[\cdot\given T=t]$ denotes the empirical mean within treatment group $T=t$ in the source data. 
One can apply our framework to the treated and control groups separately. In this case, the distribution shift model posits that the treatment assignment distribution is held fixed, and the random/hybrid shift applies to the $(X,Y)$ distribution within each treatment group. A natural approach under this model is then to combine the AIDW and AIHW estimators for the per-group mean outcome to form the final estimator for the target average treatment effect. 


\subsection{Equal-mass partitions with atoms}
\label{app:atoms}
When $P_{s,D}$ has atoms (equivalently $P_s$, since
$w_D$ is bounded away from zero), as is the case
with discrete covariates or outcomes, the
equal-mass partition in
Assumption~\ref{ass:model}(c) may not exist on the
space of $(X,Y)$. This poses no additional
difficulty. We enlarge the probability space with
an independent $U\sim\mathrm{Unif}[0,1]$ and
partition the space of $(X,Y,U)$ instead, which
admits cells of mass exactly $1/J$ for every $J$
and any Borel law; the target law of $(X,Y)$ is
then recovered by pushing forward along the
coordinate projection. An atom of mass $p$ is
thereby split across approximately $pJ$ cells, so
its target mass is rescaled by the average of the
corresponding weights rather than by a single
$W_j$. For readability, the main text 
directly on the space of $(X,Y)$.

\subsection{Deferred details on diagnostic tools}
\label{app:subsec_detail_diagnostic}

Here we present a plug-in estimator for the $s_k^2$ in Corollary~\ref{cor:model-check}. Let 
\begin{equation*}
\hat s_k^2
= (1/ n_s)\hat{\mathrm{Var}}_{P_s} \bigl(\hat w_{\mathcal D}(X_{\mathcal D})(X_k-\hat Q_k(X_{\mathcal D}))\bigr) 
+
 (1/n_t+\hat\delta_{\mathrm{dist}}^2)\hat V_{s,\mathcal D,k},
\end{equation*}
where $\hat{\mathrm{Var}}_{P_s}$ is the source sample variance and $\hat\delta_{\mathrm{dist}}^2$ is the estimator from Section~\ref{sec:estimation-of-delta}, and
\[
\hat V_{s,\mathcal D,k}
=
\frac{\sum_{i=1}^{n_s}\hat w_{\mathcal D}(X_{\mathcal D,i})(R_{ik}-\bar R_{w,k})^2}
{\sum_{i=1}^{n_s}\hat w_{\mathcal D}(X_{\mathcal D,i})},
\qquad
R_{ik}:=X_{k,i}-\hat Q_k(X_{\mathcal D,i}),
\quad
\bar R_{w,k}:=
\frac{\sum_{i=1}^{n_s}\hat w_{\mathcal D}(X_{\mathcal D,i})R_{ik}}
{\sum_{i=1}^{n_s}\hat w_{\mathcal D}(X_{\mathcal D,i})}.
\]

\section{Proofs}

\subsection{Proof of Theorem~\ref{theorem:aidw}}
\label{app:subsec_proof_aidw}
\begin{proof}
Let us first define
\begin{equation*}
    \hat \theta_{oracle} = \frac{1}{n_s} \sum_{i=1}^{n_s} (Y_i -  Q(X_i)) + \alpha \frac{1}{n_s} \sum_{i=1}^{n_s} Q(X_i)  + (1-\alpha) \frac{1}{n_t} \sum_{i=1}^{n_t}  Q(X_i').
\end{equation*}
Since $\hat Q$ is trained on auxiliary data independent of the evaluation samples, we condition on that auxiliary training sample throughout the remainder bounds below and treat $\hat Q$ as fixed.
Let $\mu_{\hat Q} = E_s[\hat Q]$ and $\mu_{Q} = E_s[Q]$. First note that in the definition of $\hat \theta_{\mathrm{AIDW}}(\hat \alpha)$ and $\hat \theta_{oracle}$ we can replace $\hat Q$ by $\hat Q - \mu_{\hat Q}$ and $Q$ by $Q - \mu_{Q}$ without changing $\hat \theta_{\mathrm{AIDW}}(\hat \alpha)$ or $\hat \theta_{oracle}$ (the contributions of any constant in the source residual term and the source/target prediction terms cancel because $\alpha+(1-\alpha)=1$). For the entire remainder analysis up to and including \eqref{eq:negligible} and the displayed bound for $\hat\theta_{\mathrm{AIDW}}(\hat\alpha)-\hat\theta_{\mathrm{oracle}}$, we therefore work with these centered versions and, with a slight abuse of notation, continue to write $\hat Q$ and $Q$; in particular, $E_s[\hat Q]=E_s[Q]=0$ throughout that block. The bounds we derive transfer immediately to the uncentered objects because $\hat \theta_{\mathrm{AIDW}}(\hat \alpha) - \hat \theta_{oracle}$ is invariant under constant shifts in $\hat Q$ and $Q$. We will now show that $\hat \theta_{\mathrm{AIDW}}(\hat \alpha) - \hat \theta_{oracle} = o_P(1/\sqrt{n_s})$.
For the target remainder, define $\Delta(X)=\hat Q(X)-Q(X)$. Since both $\hat Q$ and $Q$ have mean zero under $P_s$, we also have $E_s[\Delta]=0$. Write
\begin{align*}
  \bar\Delta_t := \frac{1}{n_t}\sum_{i=1}^{n_t}\Delta(X_i')
  = \underbrace{\left(\frac{1}{n_t}\sum_{i=1}^{n_t}\Delta(X_i') - E_t[\Delta]\right)}_{\mathrm{(A)}}
  + \underbrace{\left(E_t[\Delta]-E_s[\Delta]\right)}_{\mathrm{(B)}}.
\end{align*}
Conditionally on the perturbation and the auxiliary training sample, term $\mathrm{(A)}$ is centered, so Chebyshev's inequality gives
\begin{align*}
 \mathbb{P}(\sqrt{n_s}|\mathrm{(A)}|\ge \epsilon \mid \hat Q)
 &= \mathbb{E} \left[\mathbb{P}(\sqrt{n_s}|\mathrm{(A)}|\ge \epsilon\mid W_{\bullet},\hat Q)\,\middle|\,\hat Q\right] \\
 &\le \frac{n_s}{n_t\epsilon^2}\,\mathbb{E}[\mathrm{Var}_t(\Delta)\mid \hat Q]
 \le \frac{n_s}{n_t\epsilon^2}\,\mathbb{E}[E_t[\Delta^2]\mid \hat Q]
 = \frac{n_s}{n_t\epsilon^2}\,E_s[\Delta^2].
\end{align*}
As $n_s\asymp n_t$ and by assumption $E_s[\Delta^2]=o_P(1)$, the conditional bound on the right-hand side is $o_P(1)$. Since the conditional probability on the left is bounded by $1$, averaging over the auxiliary training sample yields $\mathrm{(A)}=o_P(n_s^{-1/2})$.

To control $\mathrm{(B)}$, define $\bar W = J^{-1}\sum_{j=1}^J W_j$ and
\[
  m_j := J\int_{I_j}\Delta(x)\,dP_s(x,y), \qquad j=1,\ldots,J.
\]
Because $P_s(I_j)=1/J$ and $E_s[\Delta]=0$, we have $J^{-1}\sum_{j=1}^J m_j=0$. Moreover,
\[
  \mathrm{(B)}
  =
  \frac{1}{J\bar W}\sum_{j=1}^J (W_j-\bar W)m_j
  =
  \frac{1}{J\bar W}\sum_{j=1}^J (W_j-1)m_j.
\]
Using $W_j\ge c>0$ and Chebyshev's inequality,
\begin{align*}
 \mathbb{P}(\sqrt{n_s}|\mathrm{(B)}|\ge \epsilon \mid \hat Q)
 &\le \mathbb{P}\left(\frac{1}{c}\left|\frac{1}{J}\sum_{j=1}^J (W_j-1)m_j\right|\ge \frac{\epsilon}{\sqrt{n_s}}\,\middle|\,\hat Q\right) \\
 &= \mathbb{P}\left(\left|\frac{1}{J}\sum_{j=1}^J \frac{W_j-1}{c}m_j\right|\ge \frac{\epsilon}{\sqrt{n_s}}\,\middle|\,\hat Q\right) \le \frac{n_s\,\mathrm{Var}(W_1)}{J^2 c^2 \epsilon^2}\sum_{j=1}^J m_j^2.
\end{align*}
By Jensen's inequality, $J^{-1}\sum_{j=1}^J m_j^2 \le E_s[\Delta^2]$, hence
\[
  \mathbb{P}(\sqrt{n_s}|\mathrm{(B)}|\ge \epsilon \mid \hat Q)
  \le \frac{n_s\,\mathrm{Var}(W_1)}{J c^2 \epsilon^2} E_s[\Delta^2]
  = o_P(1)
\]
because $n_s\asymp J$ and $E_s[\Delta^2]=o_P(1)$. Since the conditional probability on the left is bounded by $1$, averaging over the auxiliary training sample yields $\mathrm{(B)}=o_P(n_s^{-1/2})$, and we conclude that
\begin{equation}\label{eq:negligible}
	\left| \frac{1}{n_t} \sum_{i=1}^{n_t} \big(\hat Q(X_i') - Q(X_i')\big) \right| = o_P(1/\sqrt{n_s}).
\end{equation}
The source-sample analogue follows from the same conditional-Chebyshev argument: under the centering convention $E_s[\Delta]=0$,
\[
\mathbb{P} \left(\sqrt{n_s}\left|\frac{1}{n_s}\sum_{i=1}^{n_s}\Delta(X_i)\right|\ge \epsilon \,\middle|\, \hat Q\right)
\le \epsilon^{-2}E_s[\Delta^2]
= o_P(1),
\]
so $n_s^{-1}\sum_{i=1}^{n_s}(\hat Q(X_i)-Q(X_i))=o_P(n_s^{-1/2})$ after averaging over the auxiliary training sample.
Returning to the oracle replacement, we have
\begin{align*}
    \hat \theta_{\mathrm{AIDW}}(\hat \alpha) - \hat \theta_{oracle} &= \frac{1}{n_s} \sum_{i=1}^{n_s} \big(Q(X_i) - \hat Q(X_i)\big)  +  \left( \frac{1}{n_s} \sum_{i=1}^{n_s} \big(\hat Q(X_i) - Q(X_i)\big) \right) \hat \alpha \\
    &+ \left( \frac{1}{n_s} \sum_{i=1}^{n_s} Q(X_i) \right)  \left( \hat \alpha - \alpha \right) 
    +  \left( \frac{1}{n_t} \sum_{i=1}^{n_t} \big(\hat Q(X_i') - Q(X_i')\big) \right) ( 1 - \hat \alpha) \\
    &+  \left( \frac{1}{n_t} \sum_{i=1}^{n_t} Q(X_i')  \right) \left( ( 1 - \hat \alpha) - (1 - \alpha) \right) = o_P(1/\sqrt{n_s}).
\end{align*}
Here, we used equation~\eqref{eq:negligible}, the source-sample analogue just proved, the assumption $\hat \alpha - \alpha=o_P(1)$, and Appendix Theorem~\ref{thm:matched-clt} applied with baseline law $P_s$, $K=1$, and $\phi_1=Q$, which gives $\hat{\mathbb E}_t[Q]-\mathbb{E}_s[Q] = O_P((1/n_t+\delta_{\mathrm{dist}}^2)^{1/2})$; under the centering convention $\mathbb{E}_s[Q]=0$, the target empirical average is therefore $O_P(n_s^{-1/2})$ by the balanced regime $n_s\asymp n_t\asymp J$. The source empirical average $n_s^{-1}\sum_i Q(X_i)$ is also $O_P(n_s^{-1/2})$ by the ordinary CLT under the same centering. To analyze the oracle term, we now return to the original uncentered regression function $Q(x)=\mathbb E_s[Y\mid X=x]$; adding constants to $Q$ leaves $\hat\theta_{\mathrm{oracle}}-\theta$ unchanged, so the remainder bound above is unaffected. Write
\begin{align*}
  r(X,Y) &:= Y-Q(X), \quad 
  I_n  := \hat{\mathbb E}_s[r+\alpha Q]-\mathbb E_s[r+\alpha Q], \\
  S_n &:= (1-\alpha)\left(\hat{\mathbb{E}}_t[Q]-\mathbb{E}_t[Q]\right), \quad 
  D_n := \mathbb{E}_t[r+\alpha Q]-\mathbb{E}_s[r+\alpha Q].
\end{align*}
Then
\[
\hat \theta_{oracle} - \theta = I_n + S_n - D_n.
\]
We first handle the source block $I_n$, and then handle the target sampling and perturbation blocks jointly. If a variance component below is eventually zero, then the corresponding centered fluctuation is identically zero and is omitted from the normalization.
Throughout we use the identity $r+\alpha Q = Y-(1-\alpha)Q$, which holds for any fixed $Q$. By the Central Limit Theorem,
\[
v_{I,n}^{-1/2} I_n \rightarrow \mathcal{N}(0,1),
\qquad
v_{I,n}
=
\frac{1}{n_s}\mathrm{Var}_{P_s}(r+\alpha Q)
=
\frac{1}{n_s}\mathrm{Var}_{P_s}\big(Y-(1-\alpha)Q\big).
\]
We analyze the target sampling term and the perturbation term jointly by one
application of Appendix Theorem~\ref{thm:matched-clt}. Let 
$f_1:=(1-\alpha)Q$, 
$f_2:=r+\alpha Q$. 
Apply the theorem with baseline law $P_s$, two coordinates, the same
perturbation weights in both coordinates, and an unused independent empirical
coordinate for $f_2$. This auxiliary empirical coordinate is introduced only
to obtain the joint perturbation limit for $f_2$; it is discarded because
target outcomes are unobserved. Projecting the theorem's four-coordinate limit gives
the joint Gaussian limit of
\[
\sqrt J\left(
\hat{\mathbb E}_t[f_1]-\mathbb E_s[f_1],\,
\mathbb E_t[f_1]-\mathbb E_s[f_1],\,
\mathbb E_t[f_2]-\mathbb E_s[f_2]
\right).
\]
The linear combination with coefficients $(1,-1,-1)$ is
$\sqrt J(S_n-D_n)$. In the covariance matrix of
Theorem~\ref{thm:matched-clt}, the sampling difference
$\hat{\mathbb E}_t[f_1]-\mathbb E_t[f_1]$ is asymptotically uncorrelated
with the perturbation coordinate $\mathbb E_t[f_2]-\mathbb E_s[f_2]$.
Hence
\[
\sqrt J(S_n-D_n)
\Rightarrow
\mathcal N \left(0,
\frac{(1-\alpha)^2\mathrm{Var}_{P_s}(Q)}{\rho_t}
+
\mathrm{Var}(W_1)\mathrm{Var}_{P_s}(r+\alpha Q)
\right),
\qquad
\rho_t:=\lim_{J\to\infty}\frac{n_t}{J}.
\]
Set
\[
v_{S,n}
=
(1-\alpha)^2 \frac{1}{n_t}\mathrm{Var}_{P_s}(Q),
\qquad
v_{D,n}
=
\delta_\text{dist}^2 \mathrm{Var}_{P_s}(r+\alpha Q).
\]
Equivalently, $v_{D,n} = \delta_\text{dist}^2 \mathrm{Var}_{P_s}\big(Y-(1-\alpha)Q\big)$.
Since $J/n_t\to 1/\rho_t$ and $\delta_\text{dist}^2=J^{-1}\Var(W_1)$,
Slutsky's theorem gives
\begin{equation*}
   (v_{S,n}+v_{D,n})^{-1/2}(S_n-D_n)  \rightarrow \mathcal{N}(0,1).
\end{equation*}
If one of $v_{S,n}$ or $v_{D,n}$ is eventually zero, the same conclusion follows from the same display with that zero-variance component omitted.
Since the source sample is independent of the target covariate sample and the perturbation draw, the marginal limits combine into a joint limit in which the two quantities
\[
v_{I,n}^{-1/2}I_n
\quad\text{and}\quad
(v_{S,n}+v_{D,n})^{-1/2}(S_n-D_n),
\]
are independent standard-normal coordinates asymptotically. Applying Cramér--Wold and Slutsky's theorem to their deterministic linear combination gives
\begin{equation*}
  \left( v_{I,n}+v_{S,n}+v_{D,n} \right)^{-1/2}   (\hat \theta_{\mathrm{oracle}} - \theta  ) \rightarrow \mathcal{N} \left(0, 1 \right).
\end{equation*}
Since $v_{I,n}+v_{S,n}+v_{D,n}=s_n^2$, the oracle-replacement step shown above transfers the same limit to $\hat \theta_{\mathrm{AIDW}}(\hat \alpha)$.
This completes the proof.
\end{proof}

\subsection{Proof of Proposition~\ref{prop:exact_balance_benchmark}}
\label{app:subsec_proof_balance}

\begin{proof}
Because \eqref{eq:exact_balance_basis} holds and $Q_0(X)=\beta_0^\top b(X)$ lies in the balanced span,
\[
\hat{\mathbb{E}}_s[\hat w(X)Q_0(X)]
=
\beta_0^\top \hat{\mathbb{E}}_s[\hat w(X)b(X)]
=
\beta_0^\top \hat{\mathbb{E}}_t[b(X)]
=
\hat{\mathbb{E}}_t[Q_0(X)].
\]
Using $Y=Q_0(X)+\varepsilon$, we obtain
\begin{align*}
\hat\theta_{\mathrm{bal}}
&=
\hat{\mathbb{E}}_s[\hat w(X)\{Y-\hat\beta^\top b(X)\}]
+
\hat{\mathbb{E}}_t[\hat\beta^\top b(X)] \\
&=
\hat{\mathbb{E}}_s[\hat w(X)Y]
-
\hat\beta^\top \hat{\mathbb{E}}_s[\hat w(X)b(X)]
+
\hat\beta^\top \hat{\mathbb{E}}_t[b(X)] \\
&=
\hat{\mathbb{E}}_s[\hat w(X)Y] \\
&=
\hat{\mathbb{E}}_s[\hat w(X)Q_0(X)]
+
\hat{\mathbb{E}}_s[\hat w(X)\varepsilon] \\
&=
\hat{\mathbb{E}}_t[Q_0(X)]
+
\hat{\mathbb{E}}_s[\hat w(X)\varepsilon].
\end{align*}
Subtracting
\[
\theta = \mathbb{E}_t[Y] = \mathbb{E}_t[Q_0(X)] + \mathbb{E}_t[\varepsilon]
\]
gives
\[
\hat\theta_{\mathrm{bal}} - \theta
=
\big\{\hat{\mathbb{E}}_t[Q_0(X)] - \mathbb{E}_t[Q_0(X)]\big\}
+
\hat{\mathbb{E}}_s[\hat w(X)\varepsilon]
-
\mathbb{E}_t[\varepsilon].
\]
Write
\[
A := \hat{\mathbb{E}}_t[Q_0(X)] - \mathbb{E}_t[Q_0(X)],
\qquad
B := \hat{\mathbb{E}}_s[\hat w(X)\varepsilon],
\qquad
C := \mathbb{E}_t[\varepsilon],
\]
so that $\hat\theta_{\mathrm{bal}}-\theta=A+B-C$. Let
\[
\mathcal{G}
:=
\sigma \left(X_1,\ldots,X_{n_s},X_1',\ldots,X_{n_t}',\hat w,P_t\right).
\]
Because $\hat w$ depends only on covariates and $\varepsilon_i$ is independent of $X_i$ with mean zero,
\[
\mathbb{E}[B \mid \mathcal{G}]
=
\frac{1}{n_s}\sum_{i=1}^{n_s} \hat w(X_i)\,\mathbb{E}[\varepsilon_i\mid \mathcal G]
=
0.
\]
Moreover, conditional on $\mathcal{G}$, the source residuals are independent with conditional variances $\sigma_\varepsilon^2$, so
\[
\mathrm{Var}(B \mid \mathcal{G})
=
\frac{\sigma_\varepsilon^2}{n_s^2}\sum_{i=1}^{n_s} \hat w(X_i)^2
=
\frac{\sigma_\varepsilon^2}{n_s}\hat{\mathbb{E}}_s[\hat w(X)^2].
\]
Therefore,
\[
\mathrm{Var}(B)
=
\mathbb{E}[\mathrm{Var}(B \mid \mathcal{G})]
=
\frac{\sigma_\varepsilon^2}{n_s}\,\mathbb{E} \left[\hat{\mathbb{E}}_s[\hat w(X)^2]\right].
\]
Similarly, conditional on $P_t$, the target covariate sample is independent of the source sample and is i.i.d.\ from the realized target covariate marginal, so the target-sampling term satisfies
\[
\mathbb{E}[A \mid P_t] = 0,
\qquad
\mathrm{Var}(A \mid P_t) = \frac{1}{n_t}\mathrm{Var}_{P_t}(Q_0(X)),
\]
hence
\[
\mathrm{Var}(A)
=
\frac{1}{n_t}\,\mathbb{E} \left[\mathrm{Var}_{P_t}(Q_0(X))\right].
\]
The cross-covariances vanish by iterated expectation. Since $A$ is $\mathcal{G}$-measurable,
\[
\mathbb{E}[AB]
=
\mathbb{E} \left[A\,\mathbb{E}[B \mid \mathcal{G}]\right]
=
0.
\]
Since $C$ is measurable with respect to $P_t \subseteq \mathcal{G}$,
\[
\mathbb{E}[BC]
=
\mathbb{E} \left[C\,\mathbb{E}[B \mid \mathcal{G}]\right]
=
0.
\]
Finally,
\[
\mathbb{E}[AC]
=
\mathbb{E} \left[C\,\mathbb{E}[A \mid P_t]\right]
=
0.
\]
Combining the variance formulas for $A$, $B$, and $C$ yields the exact identity
\begin{equation}\label{eq:exact_balance_variance_exact}
\mathrm{Var}(\hat\theta_{\mathrm{bal}}-\theta)
=
\frac{1}{n_t}\,\mathbb{E} \left[\mathrm{Var}_{P_t}(Q_0(X))\right]
+
\frac{\sigma_\varepsilon^2}{n_s}\,\mathbb{E} \left[\hat{\mathbb{E}}_s[\hat w(X)^2]\right]
+
\mathrm{Var} \left(\mathbb{E}_t[\varepsilon]\right).
\end{equation}
For the asymptotic simplifications, we use a direct cellwise calculation. Define $m_j := J\int_{I_j}\varepsilon\,dP_s$ and $\bar W := J^{-1}\sum_{j=1}^J W_j$. Since $P_s(I_j)=1/J$ and $\mathbb{E}_s[\varepsilon]=0$, we have $J^{-1}\sum_{j=1}^J m_j = 0$ and
\[
\mathbb{E}_t[\varepsilon]
=
\frac{1}{J\bar W}\sum_{j=1}^J W_j m_j
=
\frac{1}{J\bar W}\sum_{j=1}^J (W_j-1) m_j.
\]
The partition-refinement assumption gives $J^{-1}\sum_{j=1}^J m_j^2 \to \mathrm{Var}_{P_s}(\varepsilon) = \sigma_\varepsilon^2$. To handle the random normalization directly, let $U_j=W_j/\bar W-1$. Then $\sum_j U_j=0$,
\[
\mathbb{E}_t[\varepsilon]
=
\frac1J\sum_{j=1}^J U_jm_j,
\qquad
\mathbb E[U_j]=0,
\]
where the last equality follows from exchangeability and $\sum_j W_j/\bar W=J$. Hence $\mathbb E[\mathbb E_t[\varepsilon]]=0$. By exchangeability, $\mathbb E[U_j^2]$ is common across $j$ and $\mathbb E[U_iU_j]$ is common across $i\ne j$. Since $\sum_jU_j=0$,
\[
0
=
\mathbb E \textstyle{\left[ (\sum_{j=1}^J U_j )^2\right]}
=
J\mathbb E[U_1^2]+J(J-1)\mathbb E[U_1U_2],
\]
so $\mathbb E[U_1U_2]=-\mathbb E[U_1^2]/(J-1)$. Using also $\sum_jm_j=0$, we obtain the exact identity
\[
\mathrm{Var} \left(\mathbb{E}_t[\varepsilon]\right)
=
\frac{1}{J^2}\sum_{j=1}^Jm_j^2
\left\{\mathbb E[U_1^2]-\mathbb E[U_1U_2]\right\}
=
\frac{1}{J}\frac{J}{J-1}
\bigg(\frac1J\sum_{j=1}^Jm_j^2\bigg)
\mathbb E[U_1^2].
\]
Finally, $\bar W\to1$ in probability, and $(W_1/\bar W-1)^2\le C(1+W_1^2)$ because $\bar W\ge c$. Dominated convergence therefore gives $\mathbb E[U_1^2]\to\mathrm{Var}(W_1)$. Together with $\delta_{\mathrm{dist}}^2 = J^{-1}\mathrm{Var}(W_1)$, this yields
\[
\mathrm{Var} \left(\mathbb{E}_t[\varepsilon]\right)
=
\delta_{\mathrm{dist}}^2 \,\sigma_\varepsilon^2 + o(J^{-1}).
\]
For $Q_0$, which lies in $L^2(P_s)$ because $\mathbb E_s[\|b(X)\|_2^2]<\infty$ and $Q_0(X)=\beta_0^\top b(X)$, the first two moment identities follow by exchangeability of the normalized weights: for every cell $j$, $\mathbb E[W_j/\bar W]=1$, and hence $\mathbb E[\mathbb E_t[Q_0]]=\mathbb E_s[Q_0]$ and $\mathbb E[\mathbb E_t[Q_0^2]]=\mathbb E_s[Q_0^2]$. Applying the preceding centered cellwise calculation to $Q_0-\mathbb E_s[Q_0]$ gives $\mathbb E[\mathbb E_t[Q_0]^2]=\mathbb E_s[Q_0]^2 + O(J^{-1})$. Hence
\[
\mathbb{E} \left[\mathrm{Var}_{P_t}(Q_0(X))\right]
=
\mathbb{E}_s[Q_0^2] - \mathbb E[\mathbb E_t[Q_0]^2]
=
\mathrm{Var}_{P_s}(Q_0(X)) + O(J^{-1}).
\]
Since $n_t/J \to \rho_t \in (0,\infty)$, we know 
\[
\frac{1}{n_t}\,\mathbb{E} \left[\mathrm{Var}_{P_t}(Q_0(X))\right]
=
\frac{1}{n_t}\,\mathrm{Var}_{P_s}(Q_0(X)) + o(J^{-1}).
\]
Substituting these relations into \eqref{eq:exact_balance_variance_exact} shows that the leading term is \eqref{eq:exact_balance_variance_asymptotic}.


Specializing Theorem~\ref{theorem:aidw} to $Q=Q_0$ and the variance-minimizing choice $\alpha^*$ gives the leading-order oracle AIDW variance
\[
s^2_{\mathrm{AIDW},*}
=
\frac{\mathrm{Var}_{P_s}(Q_0(X))}{\,n_t + \frac{1}{1/n_s+\delta_{\mathrm{dist}}^2}\,}
+
\left(\frac{1}{n_s}+\delta_{\mathrm{dist}}^2\right)\sigma_\varepsilon^2.
\]
Proposition~\ref{prop:exact_balance_benchmark} gives the leading-order variance
\[
\mathrm{AVar}(\hat\theta_{\mathrm{bal}}-\theta)
=
\frac{1}{n_t}\,\mathrm{Var}_{P_s}(Q_0(X))
+
\left(
\frac{1}{n_s}\,\mathbb{E} \left[\hat{\mathbb{E}}_s[\hat w(X)^2]\right]
+
\delta_{\mathrm{dist}}^2
\right)\sigma_\varepsilon^2.
\]
Because the first coordinate of $b(X)$ is 1, exact balance implies 
$\hat{\mathbb{E}}_s[\hat w(X)] = 1$. 
Hence, by Jensen's inequality, 
$\hat{\mathbb{E}}_s[\hat w(X)^2]
\ge
\hat{\mathbb{E}}_s[\hat w(X)]^2
=
1$
almost surely,  
so
\[
\frac{\sigma_\varepsilon^2}{n_s}\,
\mathbb{E} \left[\hat{\mathbb{E}}_s[\hat w(X)^2]\right]
\ge
\frac{\sigma_\varepsilon^2}{n_s}.
\]
Subtracting the oracle AIDW variance formula from the exact-balancing benchmark therefore gives
\begin{align*}
\mathrm{AVar}(\hat\theta_{\mathrm{bal}}-\theta) - s^2_{\mathrm{AIDW},*}
=\,\,&
\bigg(
\frac{1}{n_t}
-
\frac{1}{\,n_t + \frac{1}{1/n_s+\delta_{\mathrm{dist}}^2}\,}
\bigg)\mathrm{Var}_{P_s}(Q_0(X))+
\frac{\sigma_\varepsilon^2}{n_s}
\left(
\mathbb{E} \left[\hat{\mathbb{E}}_s[\hat w(X)^2]\right]-1
\right),
\end{align*}
which is \eqref{eq:aidw_vs_exact_balance_gap}. Because $1/(1/n_s+\delta_{\mathrm{dist}}^2)>0$, the first displayed coefficient is strictly positive at the displayed leading-order scale. The residual-weight term is nonnegative by Jensen's inequality and is strictly positive exactly when the realized values $\hat w(X_i)$ are nonconstant across the source sample with positive probability. Thus the leading-order gap is strictly positive whenever $\mathrm{Var}_{P_s}(Q_0(X))>0$ or the realized exact-balancing weights are nonconstant across the source sample with positive probability.
\end{proof}

\subsection{Asymptotic Gaussianity of AIHW}
\label{app:aihw-gaussianity}

\begin{proof}
For any augmentation $m$, define
\begin{align}
\widetilde\theta(m)
&:=
\hat{\mathbb E}_s[w_{\mathcal D}(Y-m)]
+\hat{\mathbb E}_t[m],
\label{eq:aihw-oracle-generic}
\end{align}  We will first prove the CLT for the oracle $\widetilde\theta(m_\infty)$. Let $f_\infty(X,Y):=Y-m_\infty(X)$. Since
$m_\infty\in\mathcal M_{\mathcal D}$,
\[
\mathbb E_{s,\mathcal D}[f_\infty]
=
\mathbb E_s \left[
w_{\mathcal D}(X_{\mathcal D})
\mathbb E_s[f_\infty\mid X_{\mathcal D}]
\right]
=0.
\]
Consequently,
\[
\mathbb E_s[w_{\mathcal D}f_\infty]=\mathbb E_{s,\mathcal{D}}[f_\infty]=0.
\]
The oracle estimator $\widetilde\theta(m_\infty)$ therefore satisfies
the exact decomposition
\begin{align}
\widetilde\theta(m_\infty)-\theta
= \, &
(\hat{\mathbb E}_s-\mathbb E_s)[w_{\mathcal D}f_\infty]
+(\hat{\mathbb E}_t-\mathbb E_t)[m_\infty]
-(\mathbb E_t-\mathbb E_{s,\mathcal{D}})[f_\infty].
\label{eq:aihw-oracle-exact-decomposition}
\end{align}
For the Gaussian limit, combine a source-sample CLT with Appendix
Theorem~\ref{thm:matched-clt}, applied to
$\phi_1=m_\infty$ and $\phi_2=Y-m_\infty$ using the same perturbation
weights. This gives the joint convergence
\[
\sqrt J
\begin{pmatrix}
(\hat{\mathbb E}_s-\mathbb E_s)[w_{\mathcal D}f_\infty]\\
(\hat{\mathbb E}_t-\mathbb E_t)[m_\infty]\\
(\mathbb E_t-\mathbb E_{s,\mathcal{D}})[f_\infty]
\end{pmatrix}
\xrightarrow{d}
\mathcal N \left(0,
\begin{pmatrix}
\rho_s^{-1}\Var_{P_s}(w_{\mathcal D}f_\infty) & 0 & 0\\
0 & \rho_t^{-1}\Var_{P_{s,\mathcal D}}(m_\infty) & 0 \\
0 & 0 & \Var(W_1)\Var_{P_{s,\mathcal{D}}}(f_\infty)
\end{pmatrix}
\right).
\]
The source coordinate is independent of the other two. The zero covariance
between the last two coordinates follows directly from the covariance matrix
in Theorem~\ref{thm:matched-clt}: subtracting the perturbed population mean
from the empirical $m_\infty$ coordinate cancels its covariance with the
population perturbation of $f_\infty$. Consequently,
\[
\sqrt J\{\widetilde\theta(m_\infty)-\theta\}
\xrightarrow{d}\mathcal N(0,v_m),
\]
where
\[
v_m
:=
\rho_s^{-1}\Var_{P_s}(w_{\mathcal D}f_\infty)
+\rho_t^{-1}\Var_{P_{s,\mathcal{D}}}(m_\infty)
+\Var(W_1)\Var_{P_{s,\mathcal{D}}}(f_\infty).
\]
The sample-size limits give $JV_n(m_\infty)\to v_m$, so the nondegeneracy
condition and Slutsky's theorem yield the asserted limit.

It remains to compare the estimator directly with this oracle.
Define
\[
a:=\hat w_{\mathcal D}-w_{\mathcal D},
\qquad
\Delta:=\hat m_n-m_\infty.
\]
Direct subtraction gives
\begin{align}
\hat\theta-\widetilde\theta(m_\infty)
=\, &
\hat{\mathbb E}_s[a\{Y-\hat m_n(X)\}]
-\hat{\mathbb E}_s[w_{\mathcal D}\Delta]
+\hat{\mathbb E}_t[\Delta].
\label{eq:aihw-plug-in-remainder}
\end{align}
For the first term, its mean satisfies
\begin{align*}
\mathbb E_s[a\{Y-\hat m_n(X)\}]
&=
\mathbb E_s[a\{Y-m_\infty(X)\}]
+\mathbb E_s[a\{m_\infty(X)-\hat m_n(X)\}]\\
&=
\mathbb E_s[a\{m_\infty(X)-\hat m_n(X)\}],
\end{align*}
because $a$ is $X_{\mathcal D}$-measurable and
$m_\infty\in\mathcal M_{\mathcal D}$. Cauchy--Schwarz and the product-rate assumption in Theorem~\ref{theorem:aihw}  therefore show that this mean is
$o_P(n_s^{-1/2})$. Conditional on the
auxiliary fits and $W_\bullet$, the main source sample is i.i.d.\ from
$P_s$. Uniform boundedness of $a$ and $\hat m_n$, the
$L^2(P_s)$-consistency of $a$, and $\mathbb E_s[Y^4]<\infty$ imply
\[
\mathbb E_s[a^2\{Y-\hat m_n(X)\}^2]
\le
\|a\|_{L^4(P_s)}^2
\|Y-\hat m_n(X)\|_{L^4(P_s)}^2
=o_P(1).
\]
Here
$\|a\|_{L^4(P_s)}^2\le\|a\|_\infty\|a\|_{L^2(P_s)}=o_P(1)$, while
boundedness of $\hat m_n$ and the fourth-moment assumption make the other
factor $O_P(1)$.
Conditional Chebyshev therefore shows that $ |\hat{\mathbb E}_s[a\{Y-\hat m_n(X)\}]| =  o_P(n_s^{-1/2})$.

For the remaining two terms in~\eqref{eq:aihw-plug-in-remainder}, use
$\mathbb E_s[w_{\mathcal D}\Delta]=\mathbb E_{s,\mathcal{D}}[\Delta]$ to write
\begin{align*}
-\hat{\mathbb E}_s[w_{\mathcal D}\Delta]
+\hat{\mathbb E}_t[\Delta]
=\;&
-(\hat{\mathbb E}_s-\mathbb E_s)[w_{\mathcal D}\Delta]
+(\hat{\mathbb E}_t-\mathbb E_t)[\Delta]
+(\mathbb E_t-\mathbb E_{s,\mathcal{D}})[\Delta].
\end{align*}
The first term is $(o_P(n_s^{-1/2})$ by conditional Chebyshev,
boundedness of $w_{\mathcal D}$, and
$\|\Delta\|_{L^2(P_s)}=o_P(1)$. For the target-sampling term, condition
first on $(\Delta,W_\bullet)$. The target observations are then i.i.d.\
from $P_t$, so
\[
\mathbb{E}  \left[
n_s\{(\hat{\mathbb E}_t-\mathbb E_t)\Delta\}^2
\, \middle| \, \Delta,W_\bullet
\right]
=
\frac{n_s}{n_t}\Var_{P_t}(\Delta)
\le
\frac{n_s}{n_t}\mathbb E_t[\Delta^2].
\]
It remains to average the last quantity over the perturbation weights. Put
$r_j:=\mathbb E_{s,\mathcal{D}}[\Delta^2\mid (X,Y)\in I_j]$. Then
\[
\mathbb E_t[\Delta^2]
=
\frac{1}{J}\sum_{j=1}^J\frac{W_j}{\bar W}r_j.
\]
Because $\Delta$ is independent of $W_\bullet$, the weights remain
exchangeable conditional on $\Delta$. Moreover,
$\sum_jW_j/\bar W=J$, so exchangeability implies
$\mathbb E[W_j/\bar W]=1$ for every $j$. Consequently,
\[
\mathbb E \left[\mathbb E_t[\Delta^2]\mid\Delta\right]
=
\frac{1}{J}\sum_{j=1}^Jr_j
=
\mathbb E_{s,\mathcal{D}}[\Delta^2].
\]
Combining the last three displays gives
\[
\mathbb E \left[
n_s\{(\hat{\mathbb E}_t-\mathbb E_t)\Delta\}^2
\, \middle| \, \Delta
\right]
\le
\frac{n_s}{n_t}\mathbb E_{s,\mathcal{D}}[\Delta^2]
=o_P(1).
\]
Chebyshev's inequality therefore yields
$(\hat{\mathbb E}_t-\mathbb E_t)\Delta=o_P(n_s^{-1/2})$.
Finally, centering $\Delta$ under $P_{s,\mathcal{D}}$, using its independence from
$W_\bullet$, and repeating the cellwise calculation gives
\[
\mathbb E \left[
n_s\{(\mathbb E_t-\mathbb E_{s,\mathcal{D}})\Delta\}^2
\, \middle| \, \Delta
\right]
\le
\frac{\Var(W_1) n_s}{c^2J}
\Var_{P_{s,\mathcal{D}}}(\Delta)
=o_P(1).
\]
All three terms are therefore $o_P(n_s^{-1/2})$, which proves oracle
replacement. Combining this fact with the oracle expansion and Gaussian limit
completes the proof.
\end{proof}

\subsection{MSE-optimal augmentation for AIHW}
\label{app:mse-optimal-aihw}

\begin{proof}
Let
\[
\varepsilon:=Y-Q(X),
\qquad
Z:=Q(X)-Q_{\mathcal D}(X_{\mathcal D}).
\]
Since $dP_{s,\mathcal D}=w_{\mathcal D}(X_{\mathcal D})\,dP_s$, conditioning under
$P_{s,\mathcal D}$ or $P_s$ gives the same law given $X_{\mathcal D}$. Hence
\[
\mathbb E_{s,\mathcal D}[\varepsilon\mid X]=0,
\qquad
\mathbb E_{s,\mathcal D}[Z\mid X_{\mathcal D}]=0.
\]
Every $m\in\mathcal M_{\mathcal D}$ has the unique representation
\[
m=Q_{\mathcal D}+u,
\qquad
\mathbb E_{s,\mathcal D}[u\mid X_{\mathcal D}]=0,
\]
and conversely every such $u$ is feasible. Moreover,
$Y-m=\varepsilon+Z-u$ and
$\mathbb E_{s,\mathcal D}[Y-m]=\mathbb E_s[w_{\mathcal D}(Y-m)]=0$.
The conditional mean-zero identities therefore
give
\begin{align*}
\Var_{P_s}\{w_{\mathcal D}(Y-m)\}
&=\mathbb E_{s,\mathcal D}[w_{\mathcal D}\varepsilon^2]
+\mathbb E_{s,\mathcal D}[w_{\mathcal D}(Z-u)^2],\\
\Var_{P_{s,\mathcal D}}(m)
&=\Var_{P_{s,\mathcal D}}(Q_{\mathcal D})+\mathbb E_{s,\mathcal D}[u^2],\\
\Var_{P_{s,\mathcal D}}(Y-m)
&=\mathbb E_{s,\mathcal D}[\varepsilon^2]+\mathbb E_{s,\mathcal D}[(Z-u)^2].
\end{align*}
Set
\[
a(X_{\mathcal D})
:=
\frac{w_{\mathcal D}(X_{\mathcal D})}{n_s}
+\delta_{\mathrm{dist}}^2,
\qquad
b:=\frac{1}{n_t}.
\]
The preceding identities reduce the objective to
\[
V_n(Q_{\mathcal D}+u)
=C_n+\mathbb E_{s,\mathcal D}[a(Z-u)^2+bu^2],
\]
where
\[
C_n
:=
\frac{1}{n_s}\mathbb E_{s,\mathcal D}[w_{\mathcal D}\varepsilon^2]
+\frac{1}{n_t}\Var_{P_{s,\mathcal D}}(Q_{\mathcal D})
+\delta_{\mathrm{dist}}^2\mathbb E_{s,\mathcal D}[\varepsilon^2]
\]
does not depend on $u$. Since $b>0$, $a+b>0$, and completing the
square pointwise gives
\[
a(Z-u)^2+bu^2
=
(a+b)\left(u-\frac{a}{a+b}Z\right)^2
+\frac{ab}{a+b}Z^2.
\]
Thus the unique candidate is
\[
u^*(X)
=
\frac{a(X_{\mathcal D})}{a(X_{\mathcal D})+b}Z(X).
\]
It is feasible because $a/(a+b)$ is $X_{\mathcal D}$-measurable:
\[
\mathbb E_{s,\mathcal D}[u^*\mid X_{\mathcal D}]
=
\frac{a}{a+b}\mathbb E_{s,\mathcal D}[Z\mid X_{\mathcal D}]
=0.
\]
For every feasible $u$, the same identity gives the exact gap
\[
V_n(Q_{\mathcal D}+u)-V_n(Q_{\mathcal D}+u^*)
=
\mathbb E_{s,\mathcal D}[(a+b)(u-u^*)^2]
\ge0,
\]
with equality only if $u=u^*$ $P_{s,\mathcal D}$-almost surely. Substitution gives equation~\eqref{eq:lambda-star}. Under the convention
$m_{\alpha_{\mathcal D}}=Q_{\mathcal D}+ (1-\alpha_{\mathcal D})(Q-Q_{\mathcal D})$, the corresponding coefficient
is $\alpha_{\mathcal D}^*=b/(a+b)$. Finally, if
$\mathbb E_{s,\mathcal D}[Z^2\mid X_{\mathcal D}]>0$, this coefficient and hence the
pooling function itself is uniquely identified.
\end{proof}

\subsection{Proof of Corollary~\ref{cor:model-check}}
\label{app:proof-model-check}

\begin{proof}
Because $\mathbb E_s[\psi_k(X)\mid X_{\mathcal D}]=0$, the choice
$\widetilde Y\equiv0$ and $m_\infty=-\psi_k$ satisfies
\[
\mathbb E_s[\widetilde Y-m_\infty(X)\mid X_{\mathcal D}]=0,
\]
so $m_\infty\in\mathcal M_{\mathcal D}$. With
$\hat m_n=-\hat\psi_k$, the estimator in
Theorem~\ref{theorem:aihw} is exactly
\[
\hat{\mathbb E}_s
 \left[\hat w_{\mathcal D}(X_{\mathcal D})\hat\psi_k(X)\right]
-\hat{\mathbb E}_t[\hat\psi_k(X)]
=\hat D_k,
\]
and its target parameter is zero. Moreover, direct substitution gives
\[
V_n(m_\infty)
=
\frac{1}{n_s}\Var_{P_s}
 \left\{w_{\mathcal D}(X_{\mathcal D})\psi_k(X)\right\}
+\left(\frac{1}{n_t}+\delta_{\mathrm{dist}}^2\right)
\Var_{P_{s,\mathcal D}}\{\psi_k(X)\}
=s_k^2.
\]
The CLT in Theorem~\ref{theorem:aihw} therefore yields
$\hat D_k/s_k\xrightarrow{d}\mathcal N(0,1)$.
\end{proof}

\subsection{Distributional CLT}\label{subsec:distributional-clt-matched}

We  now analyze the regime where the number of perturbation cells $J$ and the sample sizes $n_k$ grow at the same rate. Let $(I_j)_{j=1}^J$ be a measurable partition of the ambient sample space such that $\mathbb P^0(I_j)=1/J$ for all $j$. For fixed finite $K$, and with a perturbation-weight law that does not change with $J$, let $(W_j^1,\ldots,W_j^K)$ be a positive random vector with $\mathbb{E}[W_j^k]=1$ and $\mathrm{Var}(W_j^k)<\infty$, and assume these vectors are i.i.d.\ across $j$. This mean-one normalization is without loss of generality because the law $\mathbb P^k$ depends only on $W_j^k/\bar W^k$. Dependence across the population index $k$ is allowed within a cell and is summarized by the covariance matrix $\Sigma^W$ below. Define $\bar W^k := J^{-1}\sum_{j=1}^J W_j^k$ and, for $(x,y)\in I_j$,
\[
\frac{d\mathbb P^k}{d\mathbb P^0}(x,y)
=
\frac{W_j^k}{\bar W^k}.
\]
Conditionally on these weights, the $K$ samples are mutually independent, and the $k$th sample $\{(X_{ki},Y_{ki})\}_{i=1}^{n_k}$ is i.i.d.~from $\mathbb{P}^k$. Let $\hat{\mathbb{E}}^k[\phi]=\frac{1}{n_k}\sum_{i=1}^{n_k}\phi(X_{ki},Y_{ki})$ and retain the notation $\mathbb{E}^k[\phi]$ for the conditional expectation under $\mathbb{P}^k$. Assume furthermore that the partitions asymptotically refine $L^2(\mathbb P^0)$: for every $f\in L^2(\mathbb P^0)$,
\[
\big\|f(X,Y)-\textstyle{\sum_{j=1}^J} 1_{(X,Y)\in I_j}\,\mathbb E^0[f(X,Y)\mid (X,Y)\in I_j]\big\|_{L^2(\mathbb P^0)}
\to 0.
\]

\begin{theorem}[Distributional CLT]\label{thm:matched-clt}
Under the setup and partition-refinement assumption displayed above, suppose $\phi_k\in L^{2}(\mathbb{P}^0)$ for $k=1,\ldots,K$, and assume that $n_k\rightarrow\infty$, $J\rightarrow\infty$, with $n_k/J \rightarrow\rho_k\in(0,\infty)$ and $W_j^k \ge c$ almost surely for all $j,k,J$ and some constant $c > 0$. Let $\hat{\bm{\Phi}}=(\hat{\mathbb{E}}^1[\phi_1],\ldots,\hat{\mathbb{E}}^K[\phi_K])^{\top}$, $\bm{\Phi}=(\mathbb{E}^1[\phi_1],\ldots,\mathbb{E}^K[\phi_K])^{\top}$, and $\bm{\Phi}^0=(\mathbb{E}^0[\phi_1],\ldots,\mathbb{E}^0[\phi_K])^{\top}$. Write $\bm{\phi}(X,Y)=(\phi_1(X,Y),\ldots,\phi_K(X,Y))^\top$. Denote by $\Sigma^W$ the distributional covariance matrix with entries $(\Sigma^W)_{ij}=\mathrm{Cov}(W_1^i,W_1^j)$, and let $C$ be the $K\times K$ matrix with entries
\[
C_{ij}:=(\Sigma^W)_{ij}\,\mathrm{Cov}_{\mathbb{P}^0}(\phi_i(X,Y),\phi_j(X,Y)).
\]
Finally set $R=\mathrm{diag}(\mathrm{Var}_{\mathbb{P}^0}(\phi_1)/\rho_1,\ldots,\mathrm{Var}_{\mathbb{P}^0}(\phi_K)/\rho_K)$. Then
\[
\sqrt{J}
\begin{pmatrix}
\hat{\bm{\Phi}}-\bm{\Phi}^0 \\
\bm{\Phi}-\bm{\Phi}^0
\end{pmatrix}
\xrightarrow{d} \mathcal{N}\left(0,
\begin{pmatrix}
C+R & C \\
C & C
\end{pmatrix}
\right).
\]
\end{theorem}

\begin{proof}
The proof follows a similar technique as in \citep{zhang2025data,jeong2024out}, but we have to keep track of the joint distribution of empirical means and randomly shifted population means.

First we consider bounded $\phi$. Let $D_J=\sqrt{J}(\bm{\Phi}-\bm{\Phi}^0)$ and $S_J=\sqrt{J}(\hat{\bm{\Phi}}-\bm{\Phi})$, so that
\[
\sqrt{J}
\begin{pmatrix}
\hat{\bm{\Phi}}-\bm{\Phi}^0 \\
\bm{\Phi}-\bm{\Phi}^0
\end{pmatrix}
=
\begin{pmatrix}
S_J+D_J \\
D_J
\end{pmatrix}.
\]
We first prove the limit for $D_J$ directly. Let
\[
m_{j,k}:=J\int_{I_j}\phi_k\,d\mathbb P^0,
\qquad
\bar m_k:=\mathbb E^0[\phi_k],
\qquad
\tilde m_{j,k}:=m_{j,k}-\bar m_k .
\]
Because $J^{-1}\sum_j\tilde m_{j,k}=0$,
\[
D_{J,k}
=
\frac{1}{\bar W^k\sqrt J}
\sum_{j=1}^J (W_j^k-1)\tilde m_{j,k}.
\]
Since $\bar W^k\to1$ in probability, it is enough by Slutsky's theorem to analyze the unnormalized triangular array. For any $a\in\mathbb R^K$, the summands
\[
\frac{1}{\sqrt J}\sum_{k=1}^K a_k(W_j^k-1)\tilde m_{j,k},
\qquad j=1,\ldots,J,
\]
are independent, centered, and satisfy Lindeberg's condition: the $\tilde m_{j,k}$'s are uniformly bounded in this bounded-$\phi$ step, while the vector $W_j-\mathbf 1$ has finite second moment. Indeed, for any fixed $\epsilon>0$, the Lindeberg event requires $\|W_j-\mathbf 1\|\gtrsim \sqrt J$, so the average truncated second moment tends to zero by the fixed-law finite second moment assumption. Their variance converges to
\[
\sum_{k,\ell=1}^K
a_ka_\ell(\Sigma^W)_{k\ell}
\left(\frac1J\sum_{j=1}^J\tilde m_{j,k}\tilde m_{j,\ell}\right)
\to
a^\top C a,
\]
because the partition-refinement assumption makes the covariance of the cellwise projections converge to
$\mathrm{Cov}_{\mathbb P^0}(\phi_k,\phi_\ell)$. The triangular-array CLT and Cramér--Wold therefore give
\[
 D_J \Rightarrow \mathcal{N}(0,\,C), \qquad C_{ij}=(\Sigma^W)_{ij}\,\mathrm{Cov}_{\mathbb{P}^0}(\phi_i(X,Y),\phi_j(X,Y)).
\]
For the sampling component, conditionally on $W$, the $K$ samples are mutually independent and the $k$th sample is i.i.d.\ from $\mathbb P^k$. Fix an arbitrary vector $t_1=(t_{11},\ldots,t_{1K})\in\mathbb R^K$. Then
\[
t_1^\top S_J
=
\sum_{k=1}^K \sum_{i=1}^{n_k}
\frac{t_{1k}\sqrt J}{n_k}
\left\{\phi_k(X_{ki},Y_{ki})-\mathbb E^k[\phi_k]\right\}.
\]
Because the $\phi_k$'s are bounded in this part of the proof, conditionally on $W$ the summands in $t_1^\top S_J$ are independent, centered, and uniformly bounded by a constant times $\max_k\sqrt J/n_k$. The conditional Lyapunov (equivalently, characteristic-function Taylor) remainder is bounded by a deterministic constant, depending on $\|\phi_k\|_\infty$ and $t_1$, times
\[
\sum_{k=1}^K n_k\left|\frac{t_{1k}\sqrt J}{n_k}\right|^3
=O(J^{-1/2}).
\]
Thus, conditional on $W$,
\[
\mathbb E \left[e^{i t_1^\top S_J}\mid W\right]
=
\exp \left\{
-\frac{1}{2}
\sum_{k=1}^K t_{1k}^2\frac{J}{n_k}\mathrm{Var}_{\mathbb P^k}(\phi_k)
+ o(1)
\right\},
\]
where the $o(1)$ is deterministic for fixed $t_1$. It remains only to replace the conditional variances. Because $\phi_k$ is bounded, the same perturbation calculation that yields the limit for $D_J$, now applied to the one-dimensional functions $\phi_k$ and $\phi_k^2$, gives
\[
\mathbb E^k[\phi_k]-\mathbb E^0[\phi_k]=O_P(J^{-1/2}),
\qquad
\mathbb E^k[\phi_k^2]-\mathbb E^0[\phi_k^2]=O_P(J^{-1/2}),
\]
and therefore
\[
\Var_{\mathbb P^k}(\phi_k)=\Var_{\mathbb P^0}(\phi_k)+o_P(1).
\]
Since $n_k/J \to \rho_k\in(0,\infty)$, the conditional variance in the preceding display satisfies
\[
\sum_{k=1}^K t_{1k}^2\frac{J}{n_k}\mathrm{Var}_{\mathbb P^k}(\phi_k)
=
t_1^\top R t_1+o_P(1).
\]
Therefore
\[
A_J(t_1)
:=
\mathbb{E} \left[ e^{ i t_1^{\top} S_J }\,\middle|\, W\right]
\to
a(t_1)
:=
e^{-\frac{1}{2} t_1^{\top} R t_1}
\quad\text{in probability.}
\]
Because $|A_J(t_1)|\le 1$, this convergence also holds in $L^1$. It remains to establish the joint limit. For any $t_1,t_2\in\mathbb{R}^K$,
\begin{align*}
\mathbb{E} \left[\exp \left\{ i t_1^{\top}(S_J+D_J) + i t_2^{\top}D_J \right\}\right]
=\,&
\mathbb{E} \left[\big(A_J(t_1)-a(t_1)\big)e^{ i (t_1+t_2)^{\top} D_J } \right] +
a(t_1)\,\mathbb{E} \left[e^{ i (t_1+t_2)^{\top} D_J } \right].
\end{align*}
The first term converges to zero because
\[
\left|\mathbb{E} \left[\big(A_J(t_1)-a(t_1)\big)e^{ i (t_1+t_2)^{\top} D_J } \right]\right|
\le
\mathbb E \left[|A_J(t_1)-a(t_1)|\right]
\to 0,
\]
while the second converges to
\[
e^{-\frac{1}{2} t_1^{\top} R t_1}\, e^{-\frac{1}{2} (t_1+t_2)^{\top} C (t_1+t_2)}
\]
by the distributional CLT for $D_J$. The limit is the characteristic function of a centered Gaussian vector with covariance
\[
\begin{pmatrix}
C+R & C \\
C & C
\end{pmatrix}.
\]
This proves the claim for bounded $\phi_k$.

Extension to square-integrable $\phi$. Fix $k$ and let the truncation $T_B(x)=\max(\min(x,B),-B)$. Define the bounded, mean-matched functions
\[
 \phi_k^{(B)} := T_B(\phi_k) - \mathbb{E}^{0}[T_B(\phi_k)] + \mathbb{E}^{0}[\phi_k],\quad k=1,\ldots,K,
\]
so that $\mathbb{E}^{0}[\phi_k^{(B)}]=\mathbb{E}^{0}[\phi_k]$ and $\|\phi_k^{(B)}-\phi_k\|_{L^2(\mathbb{P}^0)}\to 0$ as $B\to\infty$. Write $\psi_k^{(B)}:=\phi_k-\phi_k^{(B)}$ and collect vectors with a bold symbol. Then
\[
\mathcal{R}_J^{(B)}
:= \sqrt{J} \begin{pmatrix}
\hat{\bm{\Phi}}-\hat{\bm{\Phi}}^{(B)} \\
\bm{\Phi}-\bm{\Phi}^{(B)}
\end{pmatrix}
=
\begin{pmatrix}
S_J(\bm{\psi}^{(B)}) + D_J(\bm{\psi}^{(B)}) \\
D_J(\bm{\psi}^{(B)})
\end{pmatrix},
\]
where $S_J(\cdot),D_J(\cdot)$ denote the sampling and distributional parts applied componentwise. We claim that
\begin{equation}\label{eq:l2_remainder_bound}
\lim_{B\to\infty}\sup_J \mathbb{E}\big[\|\mathcal{R}_J^{(B)}\|_2^2\big] = 0.
\end{equation}
It is enough to prove the following uniform $L^2$-continuity bound: for any mean-zero $\bm\psi=(\psi_1,\ldots,\psi_K)\in L^2(\mathbb P^0)^K$,
\[
\mathbb E\|S_J(\bm\psi)\|_2^2+\mathbb E\|D_J(\bm\psi)\|_2^2
\le
C\sum_{k=1}^K\|\psi_k\|_{L^2(\mathbb P^0)}^2
\]
with $C$ independent of $J$ and $\bm\psi$, after enlarging $C$ to cover finitely many initial $J$'s. Indeed, conditionally on $W$,
\[
\mathbb E[\|S_J(\bm\psi)\|_2^2\mid W]
=
\sum_{k=1}^K\frac{J}{n_k}\Var_{\mathbb P^k}(\psi_k),
\]
and exchangeability plus $\sum_j W_j^k/\bar W^k=J$ gives
$\mathbb E[\mathbb E^k[\psi_k^2]]=\mathbb E^0[\psi_k^2]$. Since $J/n_k$ is bounded, the sampling part is bounded by a constant times $\sum_k\|\psi_k\|_2^2$.
For the distributional part, write
\[
D_{J,k}(\psi_k)
=
\frac{1}{\sqrt J}\sum_{j=1}^J
\left(\frac{W_j^k}{\bar W^k}-1\right)m_{j,k},
\qquad
m_{j,k}:=J\int_{I_j}\psi_k\,d\mathbb P^0 .
\]
Because $\mathbb E^0[\psi_k]=0$, $J^{-1}\sum_jm_{j,k}=0$. Thus, using $\bar W^k\ge c$,
\[
\mathbb E[D_{J,k}(\psi_k)^2]
\le
\frac{\Var(W_1^k)}{c^2J}\sum_{j=1}^Jm_{j,k}^2
\le
\frac{\Var(W_1^k)}{c^2}\|\psi_k\|_{L^2(\mathbb P^0)}^2,
\]
where the first inequality uses $\sum_jm_{j,k}=0$, so $\sum_j(W_j^k-\bar W^k)m_{j,k}=\sum_j(W_j^k-1)m_{j,k}$, and the last inequality is Jensen's inequality within cells. This proves the continuity bound. Applying it to $\bm\psi^{(B)}$ and using the bound
\[
\|\mathcal{R}_J^{(B)}\|_2^2
\le 2\|S_J(\bm{\psi}^{(B)})\|_2^2 + 3\|D_J(\bm{\psi}^{(B)})\|_2^2
\]
gives \eqref{eq:l2_remainder_bound}, since $\|\psi_k^{(B)}\|_{L^2(\mathbb P^0)}\to0$ for each $k$.

For each fixed $B$, the bounded case implies
\[
 \sqrt{J}
 \begin{pmatrix}
  \hat{\bm{\Phi}}^{(B)}-\bm{\Phi}^{0} \\
  \bm{\Phi}^{(B)}-\bm{\Phi}^{0}
 \end{pmatrix}
 \Rightarrow \mathcal{N} \left(0,
 \begin{pmatrix}
  C^{(B)}+R^{(B)} & C^{(B)} \\
  C^{(B)} & C^{(B)}
 \end{pmatrix}\right),
\]
with $(C^{(B)})_{ij}=(\Sigma^W)_{ij}\,\mathrm{Cov}_{\mathbb{P}^0}(\phi_i^{(B)}(X,Y),\phi_j^{(B)}(X,Y))$ and $R^{(B)}=\mathrm{diag}(\mathrm{Var}_{\mathbb{P}^0}(\phi_1^{(B)})/\rho_1,\ldots)$. By $L^2$ convergence and Cauchy--Schwarz, $C^{(B)}\to C$ and $R^{(B)}\to R$ as $B\to\infty$. The uniform remainder bound \eqref{eq:l2_remainder_bound} and the converging-together theorem therefore yield, for all $\phi_k\in L^2(\mathbb{P}^0)$,
\[
 \sqrt{J}
 \begin{pmatrix}
  \hat{\bm{\Phi}}-\bm{\Phi}^{0} \\
  \bm{\Phi}-\bm{\Phi}^{0}
 \end{pmatrix}
 \Rightarrow \mathcal{N} \left(0,
 \begin{pmatrix}
  C+R & C \\
  C & C
 \end{pmatrix}\right).
\]
This completes the proof.
\end{proof}

\begin{corollary}[Hybrid discrepancy for mean-matched covariate functions]\label{cor:hybrid_phi_discrepancy}
Consider the hybrid shift model of Section~\ref{sec:aihw}. Assume the same partition-refinement and fixed-law perturbation-weight conditions as in Theorem~\ref{thm:matched-clt}, with baseline law $P_{s,\mathcal D}$, and assume $n_s/J\to\rho_s\in(0,\infty)$, $n_t/J\to\rho_t\in(0,\infty)$. The source sample is independent of the target covariate sample and perturbation draw. Write $\delta_{\mathrm{dist}}^2=J^{-1}\mathrm{Var}(W_1)$. Let $\phi(X)$ be a measurable function of the covariates such that $\phi\in L^2(P_s)\cap L^2(P_{s,\mathcal D})$ and
\[
\mathbb{E}_s[\phi(X)] = 0,
\qquad
\mathbb{E}_{s,\mathcal D}[\phi(X)] = \mathbb{E}_s[w_{\mathcal D}(X_{\mathcal D})\phi(X)] = 0.
\]
Define
\[
s_{\phi,n}^2
:=
\frac{1}{n_s}\Var_{P_s}(\phi(X))
+
\left(\frac{1}{n_t}+\delta_{\mathrm{dist}}^2\right)\Var_{P_{s,\mathcal D}}(\phi(X)).
\]
If $s_{\phi,n}^2$ is eventually positive, then, under the joint law that averages over the source sample, the target covariate sample, and the perturbation draw,
\[
s_{\phi,n}^{-1}\Big(\hat{\mathbb E}_s[\phi(X)]-\hat{\mathbb E}_t[\phi(X)]\Big)
\xrightarrow{d}
\mathcal N(0,1).
\]
\end{corollary}

\begin{proof}
Write
\[
A_n := \hat{\mathbb E}_s[\phi(X)],
\qquad
B_n := \hat{\mathbb E}_t[\phi(X)],
\]
and abbreviate
\[
V_s := \Var_{P_s}(\phi(X)),
\qquad
V_{\mathcal D} := \Var_{P_{s,\mathcal D}}(\phi(X)).
\]
Because $\mathbb E_s[\phi(X)]=0$, the ordinary central limit theorem gives $\sqrt{n_s}A_n\Rightarrow \mathcal N(0,V_s)$. Multiplying by $\sqrt{J/n_s}\to\rho_s^{-1/2}$ gives
\[
\sqrt{J}\,A_n
\xrightarrow{d}
\mathcal N \left(0,\frac{V_s}{\rho_s}\right).
\]
For the target term, apply Theorem~\ref{thm:matched-clt} on the full joint law with $K=1$, baseline law $\mathbb P^0=P_{s,\mathcal D}$, perturbed law $\mathbb P^1=P_t$, and test function $\phi_1(x,y)=\phi(x)$. Under the hybrid model, $P_t$ is obtained by randomly perturbing the baseline law $P_{s,\mathcal D}$ across joint cells $I_1,\ldots,I_J$. Since the test function depends only on $x$, its empirical target average is exactly $B_n=\hat{\mathbb E}_t[\phi(X)]$. Because $\mathbb E_{s,\mathcal D}[\phi(X)]=0$, the first coordinate of the theorem gives
\[
\sqrt{J}\,B_n
\xrightarrow{d}
\mathcal N \left(0,\Var(W_1)V_{\mathcal D}+ V_{\mathcal D} / \rho_t \right).
\]
Here the target sampling contribution is scaled by $\sqrt{J/n_t}\to\rho_t^{-1/2}$, while the perturbation contribution is already on the $\sqrt J$ scale. The source empirical term $A_n$ is a function only of the source sample, whereas $B_n$ is a function only of the target covariate sample and perturbation draw. By the assumed independence, their characteristic functions factor for every $J$, so the off-diagonal covariance in the joint limit is exactly zero:
\[
\sqrt{J}
\begin{pmatrix}
A_n \\
B_n
\end{pmatrix}
\xrightarrow{d}
\mathcal N \left(
0,
\begin{pmatrix}
\frac{V_s}{\rho_s} & 0 \\
0 & \Var(W_1)V_{\mathcal D}+\frac{V_{\mathcal D}}{\rho_t}
\end{pmatrix}
\right).
\]
Applying the continuous map $(a,b)\mapsto a-b$ gives
\[
\sqrt{J}\,(A_n-B_n)
\xrightarrow{d}
\mathcal N \left(
0,
 V_s/\rho_s
+
  V_{\mathcal D} / \rho_t 
+
\Var(W_1)V_{\mathcal D}
\right).
\]
Finally, because 
$\delta_{\mathrm{dist}}^2
=
\frac{1}{J}\Var(W_1)$ 
is the finite-$J$ distributional variance scale, we have
\[
J s_{\phi,n}^2
=
\frac{J}{n_s}V_s
+
\left(\frac{J}{n_t}+J\delta_{\mathrm{dist}}^2\right)V_{\mathcal D}
\to
\frac{V_s}{\rho_s}
+
\frac{V_{\mathcal D}}{\rho_t}
+
\Var(W_1)V_{\mathcal D}.
\]
Therefore Slutsky's theorem yields
\[
s_{\phi,n}^{-1}(A_n-B_n)
=
\frac{\sqrt{J}(A_n-B_n)}{\sqrt{J s_{\phi,n}^2}}
\xrightarrow{d}
\mathcal N(0,1),
\]
which is the stated claim.
\end{proof}

\section{Experiment details and additional results}

\subsection{Implementation of the estimators}
\label{subsec:implementation}

We observe labeled source data $\{(X_i,Y_i)\}_{i=1}^{n_s}$ and unlabeled target covariates
$\{X_j'\}_{j=1}^{n_t}$. Below we introduce the implementation of estimators for the target mean $\EE[Y_j']$. For the ATE analyses, all estimators are applied
separately within the treated and control groups, and the final estimand is
formed by differencing the two target-mean estimates.  

All nuisance quantities are estimated by source-side cross-fitting. We first uniformly split the source data into equal-sized folds 
$\mathcal{I}_1,\dots,\mathcal{I}_K$, and let
$k(i)$ be the fold containing source unit $i$. For each fold $k$, nuisance
models are fit on the source training sample
$\{(X_i,Y_i): i \notin \mathcal{I}_k\}$ and then evaluated on the held-out
source fold $\mathcal{I}_k$. Target-side regression predictions are computed by
evaluating each fold-specific outcome model on the full target sample and then
averaging across folds.

\noindent\textbf{AIPW estimator.} The implementation of the AIPW estimator follows cross-fitting~\citep{chernozhukov2018double}. 
Let $\hat m_{-k(i)}(X_i)$ denote the cross-fitted outcome prediction for source
unit $i$, 
and let
\[
\hat\mu_T
=
\frac{1}{K}\sum_{k=1}^K \frac{1}{n_t}\sum_{j=1}^{n_t} \hat m_{-k}(X_j')
\]
be the target-average regression prediction. 

We estimate the density ratio $w(x)=p_T(x)/p_S(x)$ using a probabilistic domain classifier. For fold $k$, we train a classifier on the combined sample
$ 
\{X_i: i\notin\mathcal{I}_k\}\cup\{X_j':1\le j\le n_t\},
$
where source observations receive domain label $0$ and target observations receive domain label $1$. Let
$ 
\hat p_{-k}(x)=\widehat{\PP}(D=1\mid X=x)
$
be the fitted probability from this classifier. We then set 
$ 
\hat\pi_{S,k}=\frac{n_{S,-k}}{n_{S,-k}+n_t}
$, $
\hat\pi_{T,k}=\frac{n_t}{n_{S,-k}+n_t}$, 
and   for $i\in\mathcal{I}_k$ we set 
$\hat w_i
=
\frac{\hat p_{-k}(X_i)}{1-\hat p_{-k}(X_i)}
\frac{\hat\pi_{S,k}}{\hat\pi_{T,k}}$.  
We then apply a minimal weight
clipping step for numerical stability: 
$\hat w_i \leftarrow
\min \{\sqrt{n_s},\max (n_s^{-1/2},\hat w_i ) \}$. 
For ATE analyses, the same rule is applied within each treatment arm, replacing
$n_s$ by the source sample size in that arm.
The implemented AIPW estimator is
\[
\hat\psi_{\mathrm{AIPW}}
=
\frac{1}{n_s}\sum_{i=1}^{n_s} \hat w_i \bigl(Y_i - \hat m_{-k(i)}(X_i)\bigr)
+
\hat\mu_T.
\]

\noindent\textbf{Stable balancing weights estimator.}
The SBW baseline estimates the target mean using a weighted source outcome
average, 
    $\hat\theta_{\mathrm{SBW}}
    =
    \frac{1}{n_s}\sum_{i=1}^{n_s}\widehat\omega_iY_i$. 
The weights solve the stable balancing problem
\[
\begin{split}
    \underset{\omega_1,\ldots,\omega_{n_s}}{\operatorname{minimize}}
    \quad&
    \sum_{i=1}^{n_s}\omega_i^2,\\
    \text{subject to}\quad&
    \omega_i\geq 0,\qquad
    \frac{1}{n_s}\sum_{i=1}^{n_s}\omega_i=1,\\
    &
    \left|
        \frac{1}{n_s}\sum_{i=1}^{n_s}
        \omega_i\widetilde X_{i\ell}
        -
        \frac{1}{n_t}\sum_{j=1}^{n_t}
        \widetilde X_{j\ell}'
    \right|
    \leq \epsilon_{\mathrm{bal}}
    \quad\text{for every coordinate }\ell,
\end{split}
\]
where we set 
    $\epsilon_{\mathrm{bal}}=0.01$. 
The balancing variables are the raw covariates, standardized using source
means and standard deviations. Prespecified numerical rank reduction is
applied only to remove nearly constant or strongly collinear constraints.
The same balance features and solver settings are used in every simulation.
Note that SBW does not use an outcome regression in its point
estimator.

\noindent\textbf{AIDW estimator.}
AIDW uses only the outcome regression. With
$\hat m_i = \hat m_{-k(i)}(X_i)$ and
\[
\hat\mu_S = \frac{1}{n_s}\sum_{i=1}^{n_s} \hat m_i,
\qquad
\hat\mu_T = \frac{1}{K}\sum_{k=1}^K \frac{1}{n_t}\sum_{j=1}^{n_t}\hat m_{-k}(X_j'),
\]
the estimator is
\[
\hat\psi_{\mathrm{AIDW}}(\alpha)
=
\frac{1}{n_s}\sum_{i=1}^{n_s} (Y_i-\hat m_i)
+
\alpha \hat\mu_S
+
(1-\alpha)\hat\mu_T.
\]
When $\alpha$ is not fixed in advance, it is chosen from the estimated
distribution-shift scalar $\hat\delta_{\mathrm{dist}}^2$ via
\[
\hat\alpha
=
\frac{1}{n_t(1/n_s + \hat\delta_{\mathrm{dist}}^2)+1}.
\]
The quantity $\hat\delta_{\mathrm{dist}}^2$ is estimated from the covariate shift between the
source and target samples. 
Let $\phi(\cdot)$ denote the collection of test functions applied to these
standardized residual covariates. In the implementation, we set 
$ 
\phi(r) = \bigl(r, r^2, \sin(r), \cos(r)\bigr),
$ 
applied coordinatewise. For each test-function coordinate $\ell$, let
$\bar\phi_{S,\ell}$ and $\bar\phi_{T,\ell}$ be the source and target sample
means, and let $\hat v_{S,\ell}$ and $\hat v_{T,\ell}$ be the corresponding
sample variances. We estimate $\delta^2$ as the nonnegative solution to
\[
\frac{1}{L}
\sum_{\ell=1}^L
\frac{
(\bar\phi_{S,\ell}-\bar\phi_{T,\ell})^2
}{
(1/n_t+\delta^2)\hat v_{T,\ell} + (1/n_s)\hat v_{S,\ell}
}
= 1.
\]
If the left-hand side is already no larger than one at $\delta^2=0$, we set
$\hat\delta_{\mathrm{dist}}^2=0$. Otherwise, the root is found by bisection, with an upper cap
used only for numerical stability.

\noindent\textbf{AIHW estimator.} Our implementation follows the MSE-optimal augmentation in
Equation~\eqref{eq:lambda-star} using cross-fitting and
fold-specific variable selection. For each source fold $k$, the selected
feature set $\hat{\mathcal D}_k$ is constructed using the source
training observations $\{X_i:i\notin\mathcal I_k\}$ and the full target
covariate sample, without using the held-out source observations. On the
same source training sample, we fit a full outcome model
$\hat Q_{-k}(X)$ and a reduced outcome model
$\hat Q_{\hat{\mathcal D}_k,-k}
(X_{\hat{\mathcal D}_k})$


For the weighting step, we first fit a density-ratio model on the full
covariate vector using the  training fold (we use all the coordinates in the target samples) to obtain $\hat{w}^{\text{full}}_i$.  
We then regress the logarithm of these
weights on $X_{\hat{\mathcal D}_k}$, using the source training
observations, and exponentiate its predictions to obtain the projected
density-ratio estimator
$\hat w_{\hat{\mathcal D}_k,-k}$. The projected weights are
evaluated on both the held-out source fold and the target observations.
All weights are truncated to
$[n_s^{-1/2},n_s^{1/2}]$. In addition, when the variance of the projected
source weights exceeds that of the corresponding full-space weights, the
centered predicted log weights are shrunk until this variance bound is
satisfied.

For an observation $x$, define the fold-specific adaptive coefficient
\[
\hat\lambda_{\hat{\mathcal D}_k,-k}
(x_{\hat{\mathcal D}_k})
=
\frac{
\hat w_{\hat{\mathcal D}_k,-k}
(x_{\hat{\mathcal D}_k})/n_s
+\hat\delta_{\mathrm{dist},-k}^2
}{
1/n_t+
\hat w_{\hat{\mathcal D}_k,-k}
(x_{\hat{\mathcal D}_k})/n_s
+\hat\delta_{\mathrm{dist},-k}^2
},
\]
and the corresponding augmentation
\[
\hat m_{n,-k}(x)
=
\hat Q_{\hat{\mathcal D}_k,-k}
(x_{\hat{\mathcal D}_k})
+
\hat\lambda_{\hat{\mathcal D}_k,-k}
(x_{\hat{\mathcal D}_k})
\left\{
\hat Q_{-k}(x)
-
\hat Q_{\hat{\mathcal D}_k,-k}
(x_{\hat{\mathcal D}_k})
\right\}.
\]
The distributional distance estimator $\hat\delta_{\text{dist},-k}^2$ is detailed at the end of this part. 
The implemented cross-fitted estimator is
\[
\hat\theta_{\mathrm{AIHW}}
=
\frac{1}{n_s}
\sum_{k=1}^K\sum_{i\in\mathcal I_k}
\hat w_{\hat{\mathcal D}_k,-k}
(X_{i,\hat{\mathcal D}_k})
\{Y_i-\hat m_{n,-k}(X_i)\}
+
\frac{1}{K}\sum_{k=1}^K\frac{1}{n_t}\sum_{j=1}^{n_t}
\hat m_{n,-k}(X_j').
\]
For treatment-effect outcomes, AIHW is fitted separately
within the treatment and control groups and the two estimated target means
are differenced.
In AIHW, the distributional perturbation parameter is estimated after removing the
systematic variation explained by the selected covariates. For each fold
$k$, let $X_{-\hat{\mathcal D}_k}$ denote the covariates not included
in the selected set. Using only the source training observations, we regress
each coordinate of $X_{-\hat{\mathcal D}_k}$ on an intercept and
$X_{\hat{\mathcal D}_k}$. Let $\hat B_{-k}$ denote the resulting
coefficient matrix. We then construct the source and target residual
covariates
\[
\begin{aligned}
\hat R_{i,-k}^{S}
&=
X_{i,-\hat{\mathcal D}_k}
-
(1,X_{i,\hat{\mathcal D}_k}^{\top})\hat B_{-k},
\qquad i\notin\mathcal I_k,\\
\hat R_{j,-k}^{T}
&=
X_{j,-\hat{\mathcal D}_k}'
-
(1,X_{j,\hat{\mathcal D}_k}'^{\top})\hat B_{-k},
\qquad 1\leq j\leq n_t.
\end{aligned}
\]
Each residual coordinate is standardized using its mean and standard
deviation in the source training sample. We apply the test functions
$\phi(r)=\bigl(r,r^2,\sin(r),\cos(r)\bigr)$ 
coordinatewise to the standardized residuals. For each resulting
test-function coordinate $\ell$, let
$\bar\phi_{S,-k,\ell}$ and $\bar\phi_{T,-k,\ell}$ denote the source
and target means, and let $\hat v_{S,-k,\ell}$ and
$\hat v_{T,-k,\ell}$ denote the corresponding sample variances.

We define $\hat\delta_{\mathrm{dist},-k}^2$ as the nonnegative solution
to
\[
\frac{1}{L}\sum_{\ell=1}^{L}
\frac{
\left(\bar\phi_{S,-k,\ell}
-\bar\phi_{T,-k,\ell}\right)^2
}{
\left(1/n_t+\delta^2\right)\hat v_{T,-k,\ell}
+
\left(1/n_{s,-k}\right)\hat v_{S,-k,\ell}
}
=1,
\]
where $n_{s,-k}$ is the number of source training observations in fold
$k$. If the left-hand side is no larger than one at $\delta^2=0$, we
set $\hat\delta_{\mathrm{dist},-k}^2=0$; otherwise, the solution is
obtained by bisection. The fold-specific residualized estimate
$\hat\delta_{\mathrm{dist},-k}^2$ is then used in the corresponding
adaptive coefficient
$\hat\lambda_{\hat{\mathcal D}_k,-k}$.

\paragraph{Variable selection.}
For the AIHW estimator, we implement two heuristic selection
rules for determining $\cD$. 
First, the \texttt{gaussian-mix} rule computes the per-feature standardized
mean-difference statistic
\[
t_j
=
\frac{\bar X_{S,j}-\bar X_{T,j}}
{\sqrt{1/n_s+1/n_t}\,\hat\sigma_{S,j}},
\]
where $\bar{X}_{S,j}$ and $\bar{X}_{T,j}$ are the sample mean of the $j$-th feature in the source and target data among the data used for feature selection. We then fit a two-component Gaussian mixture model to $\{t_j\}_{j=1}^p$, and selects features in 
the component with the larger mean squared $t$-statistic.

Second, the \texttt{t-stat} rule selects feature $j$ whenever
\[
|\bar X_{S,j}-\bar X_{T,j}|
>
z_{1-\kappa/2}\sqrt{\hat\delta_{\mathrm{dist}}^2+1/n_t+1/n_s}\,\hat\sigma_{S,j},
\]
where $\kappa$ is a user-specified significance level.
If this rule selects no features, the implementation falls back to the three
features with the largest absolute $t$-statistics. Near-constant source
features are excluded.

\paragraph{Inference and variance estimation.}
In the real-data analyses, coverage is evaluated against the realized
target-sample outcome mean, $\bar Y_t=n_t^{-1}\sum_{j=1}^{n_t}Y_j'$, and 
we estimate the variance of $\hat\theta-\bar Y_t$ directly.

For the weighting baselines, we use the formula
\[
\widehat{\mathrm{se}}^2 = \widehat s_n^2 + \widehat s_{\mathrm{pred}}^2,
\]
where $\widehat s_n^2$ is the target-mean variance and
$\widehat s_{\mathrm{pred}}^2$ is the realized sample-mean variance component.

For AIPW, the implementation uses the usual cross-fitted influence-function
plug-in variance. Let 
$\hat\varphi_i^{S}
=
\hat w_i\{Y_i-\hat m_{-k(i)}(X_i)\}$ and 
$\hat\varphi_j^{T}
=
\frac{1}{K}\sum_{k=1}^K \hat m_{-k}(X_j')$. 
Then we compute the two components
\[
\widehat s^2_{n,\mathrm{AIPW}}
=
\frac{1}{n_s}\widehat{\Var}_S(\hat\varphi_i^S)
+
\frac{1}{n_t}\widehat{\Var}_T(\hat\varphi_j^T),\quad  
\widehat s^2_{\mathrm{pred,AIPW}}
=
\frac{1}{n_t}\widehat{\Var}_{\hat w}\{Y_i-\hat m_{-k(i)}(X_i)\},
\]
where $\widehat{\Var}_{\hat w}$ denotes the weighted empirical variance using
the estimated density-ratio weights. 

For SBW, let $\hat w_i$ denote the balancing weights normalized so  
$\sum_i\hat w_i=n_s$, and let $\hat\theta_{\mathrm{SBW}}
=n_s^{-1}\sum_i\hat w_iY_i$. We use
\[
\hat s_{n,\mathrm{SBW}}^2
=\frac{1}{n_s}\widehat{\operatorname{Var}}_S(\hat w_iY_i),
\qquad
\hat s_{\mathrm{pred,SBW}}^2
=\frac{1}{n_t}\frac{n_s}{n_s-1}
\frac{\sum_i\hat w_i(Y_i-\hat\theta_{\mathrm{SBW}})^2}
     {\sum_i\hat w_i},
\]
where $\widehat{\operatorname{Var}}_S$ is the ordinary sample variance. 

Let $\hat q_i=\hat Q_{-k(i)}(X_i)$ and $\hat r_i=Y_i-\hat q_i$
denote the cross-fitted outcome predictions and residuals.
For AIHW, let $\hat\phi_i=\hat m_{n,-k(i)}(X_i)$ denote the
cross-fitted adaptive augmentation defined above, and write
$\hat u_i=Y_i-\hat\phi_i$. 
Suppressing observation indices, the predictive variance estimates are
\begin{align*}
\widehat V_{\mathrm{AIDW}}
&=
\left(\frac{1}{n_s}+\frac{1}{n_t}+\hat\delta_{\mathrm{dist}}^2\right)
\widehat{\operatorname{Var}}_S
\left(\hat r+\hat\alpha\hat q\right),\\
\widehat V_{\mathrm{AIHW}}
&=
\frac{1}{n_s}\widehat{\operatorname{Var}}_S
\left(\hat w_D\hat u\right)
+
\left(\frac{1}{n_t}+\hat\delta_{\mathrm{dist}}^2\right)
\widehat{\operatorname{Var}}_{\hat w_D}(\hat u).
\end{align*}
Here $\widehat{\operatorname{Var}}_S$ is the ordinary source-sample
variance with denominator $n_s-1$, and we define
\[
\widehat{\operatorname{Var}}_w(a)
=
\frac{\sum_{i=1}^{n_s}w_i(a_i-\bar a_w)^2}
     {\sum_{i=1}^{n_s}w_i},
\qquad
\bar a_w=\frac{\sum_iw_i a_i}{\sum_iw_i}.
\]

For simulations, the coverage is computed for the target $\theta(W_\bullet)$ instead of the target-sample outcome mean, so we use $\hat{s}_{n,\text{AIPW}}$ and $\hat{s}_{n,\text{SBW}}$ introduced above to construct the predictive intervals. 
The standard errors used in the predictive intervals for AIDW and AIHW follow Theorem~\ref{theorem:aidw} and Theorem~\ref{theorem:aihw} with the same estimated plug-in quantities introduced above.
 
\subsection{Simulation details}
\label{app:simulation-details}

\noindent\textbf{Simulation implementation of the shift model.}
The data-generating process follows the two-stage decomposition   
    $P_s \to P' \to P_t$.    
The source covariates and outcomes are generated as 
    $X_i\sim\mathcal{N}_p(0,I_p)$, 
    $Y_i=Q(X_i)+\varepsilon_i$, 
    $\varepsilon_i\sim\mathcal{N}(0,1)$, 
independently for \(i=1,\ldots,n_s\). Under pure random shift, we take \(P'=P_S\). Under hybrid shift, the intermediate
law is obtained by a Gaussian location shift, where  
    $X^{\prime\prime}\sim\mathcal{N}_p(\mu_\kappa,I_p)$, $Y^{\prime\prime}=Q(X^{\prime\prime})+\varepsilon^{\prime\prime}$, 
    $\varepsilon^{\prime\prime}\sim\mathcal{N}(0,1)$, 
and we set 
    $\mu_\kappa
    = \kappa\,(v^\top,0,\ldots,0)^\top$ and 
    $v = \frac{(1,-0.75,0.5)^\top}
         {\|(1,-0.75,0.5)^\top\|_2}$. 
Thus, when \(\kappa>0\), the deterministic component depends only on 
    $\cD_0=\{1,2,3\}$. 
The corresponding intermediate-law density ratio is 
    $w_{D_0}(x_{D_0})
    =
    \exp\left\{
        \mu_{\kappa,D_0}^{\top}x_{D_0}
        -\frac{1}{2}\|\mu_{\kappa,D_0}\|_2^2
    \right\}$. 
In the different experimental settings, we shall vary the total feature dimension $p$ and the regression functions $Q(x)$, which we detail later on.

\noindent\textbf{Random-shift strength.}  
The random perturbations from $P'$ to $P_t$ are generated following~\cite{jeong2022calibrated}, where across the two shift models, we parameterize the the strength of random perturbations by
\[
    \tau=n_t\delta_{\mathrm{dist}}^2. 
\]

\noindent\textbf{Pure random-shift outcome models.}
For the pure random-shift experiments in Figure~\ref{fig:simu_aidw_linear}, we set 
    $n_s=n_t=1000$, $p=10$, $\mu_\kappa=0$, 
and vary the scaled random-shift strength over 
    $\tau\in\{0,0.5,1,2,4\}$. 
We consider two outcome regression functions. The linear setting uses 
    $Q_{\mathrm{lin}}(x)=x_1$. 
The nonlinear setting uses 
    $Q_{\mathrm{mix}}(x)
    =
    x_1
    +0.75\,\mathbf{1}\{x_2>0\}
    +0.5\,x_3\mathbf{1}\{x_3>0\}$.   
The linear outcome regression is estimated by ordinary least squares. The
nonlinear regression is estimated using 
    $\texttt{RandomForestRegressor}
     (
        \texttt{n\_estimators}=200, 
        \texttt{min\_samples\_leaf}=5
    )$, 
with all other tuning parameters fixed at their default values. 

\noindent\textbf{Linear hybrid-shift experiment.}
The experiments in Figure~\ref{fig:simu_aihw_core} uses \(p=10\) and 
    $Q_{\mathrm{hyb,lin}}(x)
    =
    x_1+0.7x_2+0.5x_3+0.8x_4+0.5x_5$.  
The two strength-varying settings in the hybrid model uses (i) $\tau=0.5$ and $\kappa\in\{0,0.5,1,1.5\}$ (deterministic-vary) and (ii) $\kappa=0.75$ and $\tau\in\{0,0.5,1,2,4\}$ (random-vary), where $\kappa$ and $\tau$ are introduced in earlier paragraphs.   
All linear outcome regressions, including the full regression \(Q\) and the
reduced regression \(Q_{D}\), are estimated by ordinary least squares.

\noindent\textbf{Dimension-varying experiments.}
For Figure~\ref{fig:simu_aihw_highd_linear}, we hold the shifts fixed at 
    $\kappa=0.75$, $\tau=1$, 
and vary 
    $p\in\{10,25,50\}$. 
Within every Monte Carlo replicate, a single source sample, target sample, and
latent target population are generated at \(p=50\). 
For the linear dimension experiment, the outcome model remains 
    $Q_{\mathrm{hyb,lin}}(x)
    =
    x_1+0.7x_2+0.5x_3+0.8x_4+0.5x_5$, 
with zero coefficients on \(X_6,\ldots,X_p\).  
The nonlinear dimension experiment retains the same hybrid-shift environment
and replaces only the outcome regression. Define 
    $T_k(x)=\Phi(x_k)$, 
where \(\Phi\) is the standard normal distribution function, and let 
    $c_{12}
    =
    \mathbb{E}\{\sin(\pi U_1U_2)\}
    =
    \int_0^1\int_0^1\sin(\pi uv)\,du\,dv$, 
for independent \(U_1,U_2\sim\operatorname{Unif}(0,1)\). The nonlinear
regression function is 
    $Q_{\mathrm{hyb,nl}}(x)
    = 
    1.5\left[
        \sin\{\pi T_1(x)T_2(x)\}-c_{12}
    \right] +
    \left[
        \{T_3(x)-1/2\}^2-\frac{1}{12}
    \right] +
    0.8\{T_4(x)-1/2\}
    +
    0.5\{T_5(x)-1/2\}$.   
Under the true deterministic set \(D_0=\{1,2,3\}\), the reduced regression is 
    $Q_{D_0}(x_{D_0})
    = 
    1.5\left[
        \sin\{\pi T_1(x)T_2(x)\}-c_{12}
    \right] +
    \left[
        \{T_3(x)-1/2\}^2-\frac{1}{12}
    \right]$.   
Both \(Q\) and \(Q_D\) are estimated using 
   $ \texttt{RandomForestRegressor}
     (
        \texttt{n\_estimators}=200, 
        \texttt{min\_samples\_leaf}=5
     )$, 
with the same fixed specification for all \(p\).

\noindent\textbf{Implementation of methods.}
The AIPW and SBW baselines are implemented in exactly the same way as those described in Section~\ref{subsec:implementation}. The Plug-in AIDW in Figure~\ref{fig:simu_aidw_linear} and the AIDW in Figures~\ref{fig:simu_aihw_core} and~\ref{fig:simu_aihw_highd_linear} are the same as the AIDW estimator used in real data analysis and described in Section~\ref{subsec:implementation}. The Adaptive AIHW in Figures~\ref{fig:simu_aihw_core} and~\ref{fig:simu_aihw_highd_linear} are the same as the AIHW\_t-stat estimator used in real data analysis and described in Section~\ref{subsec:implementation}. Below, we only introduce how the additional (semi-)oracle variants are implemented. 

In Figure~\ref{fig:simu_aidw_linear}, the \texttt{Oracle-delta AIDW} procedure is implemented in the same way as Plug-in AIDW, except that we use the known pooling coefficient $\alpha = 1/(1+n_t(1/n_s+\delta_{\text{dist}}^2))$ using knowledge of the data-generating process.  
In Figures~\ref{fig:simu_aihw_core} and~\ref{fig:simu_aihw_highd_linear}, we include two variants \texttt{Theory AIHW} and  \texttt{Known-D AIHW}. 
Compared with the default implementation of \texttt{Adaptive AIHW} introduced above, the distinctions are: 

\begin{itemize}
    \item \texttt{Theory AIHW} is supplied with the true deterministic set
    \(\cD=D_0\), the  distributional distance
    \(\delta_{\mathrm{dist}}^2\), and the oracle pooling function $\lambda_{\cD}^*(x_\cD)$ given by~\eqref{eq:lambda-star}. 
    The full and reduced outcome regressions and density-ratio weights are estimated in an identical way.  

    \item \texttt{Known-D AIHW} is supplied only with the true deterministic set \(\cD=D_0\). The estimation of regression functions, density ratios, distributional distance, and pooling functions is kept identical. 
\end{itemize}

The RMSE and coverage of all methods are evaluated against the realized target mean $\theta(W_\bullet) = \EE_t[Y\given W_\bullet]$ known by the sampling process. 

\subsection{Deferred results in experiments}
\label{app:subsec_add_exp}

\begin{figure}[htbp]
    \centering
    \includegraphics[width=\linewidth]{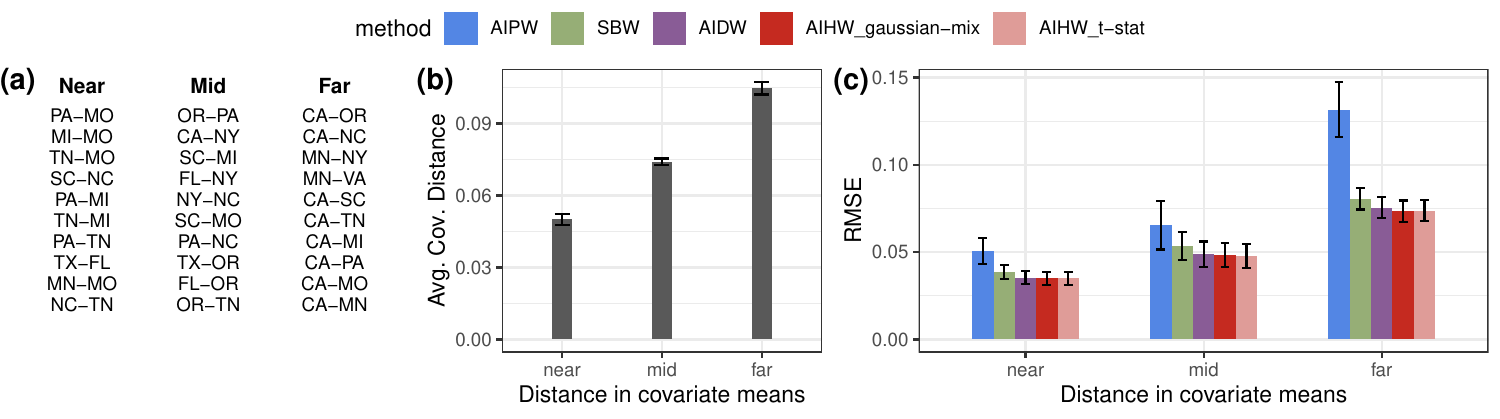}
    \caption{Variation of estimator performance with source-target covariate distance in the ACS-income data. (a) State pairs categorized as Near/Mild/Far based on covariate mean differences. (b) Average covariate mean differences in the three categories. (c) Average RMSE for state pairs in the three categories.}
    \label{fig:real_acs_by_distance}
\end{figure}